\def\bu{\bullet}
\def\marker{\>\hbox{${\vcenter{\vbox{
    \hrule height 0.4pt\hbox{\vrule width 0.4pt height 6pt
    \kern6pt\vrule width 0.4pt}\hrule height 0.4pt}}}$}\>}
\def\gpic#1{#1
     \medskip\par\noindent{\centerline{\box\graph}} \medskip}
\documentclass[12pt]{article}
\usepackage{fullpage}
\usepackage{amsmath, amssymb, latexsym, amsthm, graphicx, color}
\theoremstyle{plain}
\newtheorem{thm}{Theorem}[section]
\newtheorem{cor}[thm]{Corollary}
\newtheorem{lem}[thm]{Lemma}

\newtheorem{prop}[thm]{Proposition}
\newtheorem{ques}[thm]{Question}
\newtheorem{examp}[thm]{Example}
\theoremstyle{definition}
\newtheorem{defin}[thm]{Definition}
\newtheorem{rem}[thm]{Remark}

\newcommand{\bigceil}[1]{\big\lceil{#1}\big\rceil}

\newcommand{\lemma}[1]{Lemma~\ref{#1} } 

\DeclareMathOperator{\RS}{RS}
\DeclareMathOperator{\bin}{Bin}

\newcommand{\pr}{\mathbb{P}}
\newcommand{\ex}{\mathbb{E}}

\newcommand{\sizeof}[1]{\left\lvert #1 \right\rvert}

\newcommand{\floor}[1]{\lfloor #1 \rfloor}
\newcommand{\ceil}[1]{\lceil #1 \rceil}

\newcommand{\mtg}[1]{\overline{#1}}
\newcommand{\free}[1]{\hat{#1}}

\def\revs{revolutionaries}
\def\rev{revolutionary}
\def\floor#1{\left\lfloor #1 \right\rfloor}
\def\ceil#1{\left\lceil #1 \right\rceil}
\def\bFL#1{\bigl\lfloor #1 \bigr\rfloor}
\def\FL#1{\left\lfloor #1 \right\rfloor}
\def\CL#1{\left\lceil #1 \right\rceil}
\def\C#1{\left| #1\right |}
\def\RSG{{\rm RS}(G,m,r,s)}

\def\RS{{\rm RS}}
\def\sgmr{\sigma(G,m,r)}
\def\sgtr{\sigma(G,2,r)}
\def\sghr{\sigma(G,3,r)}
\def\VEC#1#2#3{#1_{#2},\ldots,#1_{#3}}
\def\SE#1#2#3{\sum_{#1=#2}^{#3}}
\def\SM#1#2{\sum_{#1\in #2}}
\def\nul{\varnothing}
\def\FR{\frac}
\def\NN{{\mathbb N}}

\def\CH#1#2{\binom{#1}{#2}}
\def\st{\colon\,}
\def\esub{\subseteq}
\def\us{\hat s}
\def\ds{\check s}
\def\ur{\hat r}
\def\dr{\check r}

\def\ex{\mathbb{E}}
\def\pr{\mathbb{P}}
\def\bipmtwo{\bigl\lceil{\FR{\FL{7r/2}-3}5}\bigr\rceil}
\def\bipmgen{\Bigl\lfloor{\FR12\bigl\lfloor{\FR r{\CL{m/3}}\bigr\rfloor}\Bigr\rfloor}}
\def\hispace{\phantom{\Big|}}
\def\cart{{\small\square}}

\begin{document}

\title{Revolutionaries and spies: Spy-good and spy-bad graphs}

\author{
Jane V. Butterfield\thanks{Mathematics Department, University of Illinois,
jbutter2@illinois.edu, partially supported by NSF grant DMS 08-38434,
``EMSW21-MCTP: Research Experience for Graduate Students''.}\,,
Daniel W. Cranston\thanks{Mathematics Department, Virginia Commonwealth
University, dcranston@vcu.edu.}\,,
Gregory J. Puleo\thanks{Mathematics Department, University of Illinois,
puleo@illinois.edu, partially supported by NSF grant DMS 08-38434,
``EMSW21-MCTP: Research Experience for Graduate Students''.}\,,\\
Douglas B. West\thanks{Mathematics Department, University of Illinois,
west@math.uiuc.edu, partially supported by NSA grant H98230-10-1-0363.}\,, and
Reza Zamani\thanks{Computer Science Department, University of Illinois,
zamani@uiuc.edu.}
}

\maketitle

\begin{abstract}
We study a game on a graph $G$ played by $r$ {\it \revs} and $s$ {\it spies}.
Initially, \revs\ and then spies occupy vertices.  In each subsequent round,
each \rev\ may move to a neighboring vertex or not move, and then each spy has
the same option.  The \revs\ win if $m$ of them meet at some vertex having no
spy (at the end of a round); the spies win if they can avoid this forever.

Let $\sgmr$ denote the minimum number of spies needed to win.  To avoid 
degenerate cases, assume $|V(G)|\ge r-m+1\ge\FL{r/m}\ge 1$.  The easy bounds
are then $\FL{r/m}\le \sgmr\le r-m+1$.  We prove that the lower bound is sharp
when $G$ has a rooted spanning tree $T$ such that every edge of $G$ not in $T$
joins two vertices having the same parent in $T$.  As a consequence,
$\sgmr\le\gamma(G)\FL{r/m}$, where $\gamma(G)$ is the domination number; this
bound is nearly sharp when $\gamma(G)\le m$.

For the random graph with constant edge-probability $p$, we obtain constants
$c$ and $c'$ (depending on $m$ and $p$) such that $\sgmr$ is near the trivial
upper bound when $r<c\ln n$ and at most $c'$ times the trivial lower bound when
$r>c'\ln n$.  For the hypercube $Q_d$ with $d\ge r$, we have $\sgmr=r-m+1$ when
$m=2$, and for $m\ge 3$ at least $r-39m$ spies are needed.

For complete $k$-partite graphs with partite sets of size at least $2r$, the
leading term in $\sgmr$ is approximately $\FR k{k-1}\FR rm$ when $k\ge m$.
For $k=2$, we have $\sgtr=\bipmtwo$ and $\sghr=\FL{r/2}$, and in general
$\frac{3r}{2m}-3\le \sgmr\le\frac{(1+1/\sqrt3)r}{m}$.
\looseness-1
\end{abstract}

\baselineskip16pt

\section{Introduction}
We study a pursuit game involving two teams on a graph.
The first team consists of $r$
{\it\revs}; the second consists of $s$ {\it spies}.  The \revs\ want to arrange
a one-time meeting of $m$ \revs\ free of oversight by spies.  Initially, the
\revs\ take positions at vertices, and then the spies do the same.  In each
subsequent round, each \rev\ may move to a neighboring vertex or not move, and
then each spy has the same option.  All positions are known by all players
at all times.

The \revs\ win if at the end of a round there is an {\it unguarded meeting},
where a {\it meeting} is a set of (at least) $m$ \revs\ on one vertex, and a
meeting is {\it unguarded} if there is no spy at that vertex.  The spies win if
they can prevent this forever.  Let $\RSG$ denote this game played on the graph
$G$ by $s$ spies and $r$ \revs\ seeking an unguarded meeting of size $m$.

The spies trivially win if $s\ge|V(G)|$ or $r<m$.
If $\FL{r/m}<|V(G)|$, then the \revs\ can form $\FL{r/m}$ meetings initially,
and hence at least $\FL{r/m}$ spies are needed to avoid losing immediately.
On the other hand, the spies win if $s\ge r-m+1$; they follow $r-m+1$ distinct
\revs, and the other $m-1$ \revs\ cannot form a meeting.  To avoid degenerate
or trivial games, henceforth in this paper {\bf we always assume}
$$|V(G)|\ge r-m+1\ge\FL{r/m}\ge1.$$
Let $\sgmr$ denote the minimum $s$ such that the spies win the game $\RSG$.

The game of Revolutionaries and Spies was invented by Jozef Beck in the
mid-1990s (unpublished).  Smyth promptly showed that $\sgmr=\FL{r/m}$ when $G$
is a tree, achieving the trivial lower bound (a later proof appears
in~\cite{CSW}).  Howard and Smyth~\cite{HS} studied the game when $G$ is the
infinite $2$-dimensional integer grid with one-step horizontal, vertical, and
diagonal edges.  They observed that the spy wins $\RS(G,m,2m-1,1)$ (the spy
stays at the median position), and hence $\sgmr\le r-2m+2$ when $r\ge 2m-1$
(note that always $\sgmr\le\sigma(G,m,r-1)+1$).  For $m=2$, they proved that
$6\FL{r/8}\le \sgtr\le r-2$ when $r\ge3$; they conjectured that the upper bound
is the correct answer.

Cranston, Smyth, and West~\cite{CSW} showed that $\sgmr\le\CL{r/m}$ when 
$G$ has at most one cycle.  Furthermore, let $G$ be a unicyclic graph
consisting of a cycle of length $\ell$ and $t$ vertices not on the cycle.
They showed that if $m\nmid r$ (and as usual $|V(G)|> r/m$ to avoid
degeneracies), then $\sgmr=\FL{r/m}$ if and only if
$\ell\le \max\{\FL{r/m}-t+2,3\}$.

Our objective in this paper is to advance the systematic study of this game.
We show that the trivial lower and upper bounds on $\sgmr$ each may be sharp
on various classes of graphs.  Furthermore, we obtain classes where neither
bound is asymptotically sharp and yet still $\sgmr$ can be determined or
closely approximated.

Say that $G$ is {\it spy-good} if $\sgmr$ equals the trivial lower bound
$\FL{r/m}$ for all $m$ and $r$ such that $r/m<|V(G)|$.  In
Section~\ref{spygoodsec}, we prove that every webbed tree is spy-good, where
a {\it webbed tree} is a graph $G$ containing a rooted spanning tree $T$ such
that every edge of $G$ not in $T$ joins vertices having the same parent in $T$.
For example, every graph having a dominating vertex $u$ is a webbed tree
(rooted at $u$).

Section~\ref{goodbadsec} considers general bounds.  Always
$\sgmr\le\gamma(G)\FL{r/m}$, where $\gamma(G)$ is the domination number of $G$
(the minimum size of a set $S$ such that every vertex outside $S$ has a
neighbor in $S$).  Since always $\FL{r/m}\ge (r-m+1)/m$, this upper bound is
nontrivial only when $\gamma(G)<m$.  In that case, it is nearly sharp: for
$t,m,r\in\NN$ with $t<m$, we construct a graph with domination number $t$ such
that $\sgmr>t(r/m-1)$.

In contrast to spy-good graphs, a graph $G$ is {\it spy-bad} {\it for $r$
\revs\ and meeting size $m$} if $\sgmr$ equals the trivial upper bound $r-m+1$.
Section~\ref{goodbadsec} constructs chordal graphs (and bipartite graphs)
that are spy-bad (for given $r$ and $m$).

In Section~\ref{hypersec} we study hypercubes, showing first that the
$d$-dimensional hypercube $Q_d$ is spy-bad when $d\ge r$ and $m=2$.  Also,
the winning strategy for the \revs\ uses only vertices near a fixed vertex.
By splitting the \revs\ into disjoint groups who play this strategy around
vertices far apart, it follows that when $d<r\le 2^d/d^8$, the revolutionaries
win against $(d-1)\FL{r/d}$ spies on $Q_d$ (for $m=2$).  For general $m$, we
show that hypercubes are nearly spy-bad by proving $\sigma(Q_d,m,r)\ge r-39m$
for $d\ge r\ge m$.  (For small $m$, the bound $\sigma(Q_d,m,r)\ge r-\FR34 m^2$
when $d\ge r\ge m$ is better.)

In these examples of spy-bad graphs, there are few \revs\ compared to the
number of vertices.  Similar behavior holds for the random graph with constant
edge-probability (Section~\ref{randsec}); the threshold for spies to win
depends on the relationship between $r$ and the number of vertices, $n$.
Via fairly simple arguments, we obtain constants $c$ and $c'$ (depending on
$m$) such that almost always $r-m+1$ spies are needed when $r<c\ln n$, while a
multiple of $r/m$ spies are enough when $r>c'\ln n$.
Using more intricate structural characteristics of the random graph and a more
complex strategy for the spies, Mitsche and Pra\l at~\cite{MP} independently
proved that $\sgmr=(1+o(1))r/m$ spies suffice when $r$ grows faster than
$(\log n)/p$ (here also $p$ may depend on $n$).

A complete $k$-partite graph is {\it $r$-large} if each part has at least $2r$
vertices, which is as many vertices as the players might want to use.
In Section~\ref{multisec}, we prove $\sgmr\ge \FR k{k-1}\FR rm+k$.  Also
$\sgmr\ge\FR k{k-1}\FR r{m+c}-k$ when $k\ge m$ and $c=\FR 1{k-1}$.

Section~\ref{bicliqsec} focuses on complete bipartite graphs and contains our
most delicate results.  When $G$ is an $r$-large complete bipartite graph, we
obtain $\sgtr=\bipmtwo$ and $\sghr=\FL{r/2}$.  For larger $m$ we do not have
the complete answer; we prove
$$
\left(\FR 32-o(1)\right)\FR rm-2\le \sgmr
\le \left(1+\FR1{\sqrt3}\right)\FR rm<1.58\FR rm,
$$
where the upper bound requires $\FR rm\ge\FR1{1-1/\sqrt3}$.  We conjecture that
$\sgmr$ is approximately $\frac{3r}{2m}$ when $3$ divides $m$, but in other
cases the \revs\ do a bit better.  That advantage should fade as $m$ grows,
with $\sgmr\sim \frac{3r}{2m}$.

Upper bounds for $\sgmr$ are proved using strategies for the spies.
We define a notion of \emph{stable position} in the game.  Proving that a
particular number of spies can win involves showing that in a stable position
all meetings are guarded and that for any move by the \revs\ from a stable
position, the spies can reestablish stability.  This technique is used for
graphs with dominating vertices and for webbed trees in
Section~\ref{spygoodsec}, for random graphs in Section~\ref{randsec}, and for
complete multipartite and complete bipartite graphs in Sections~\ref{multisec}
and~\ref{bicliqsec}.  Each setting uses its own definition of stability
tailored to the graphs under study.

Lower bounds are proved by strategies for the \revs, which usually are much
simpler.  Most of our winning strategies for \revs\ take at most two rounds,
but on hypercubes they take $m-1$ rounds.  In~\cite{CSW}, strategies for \revs\
proving that $\sigma(C_n,m,r)=\CL{r/m}$ (when $r/m<n$) may take many rounds.

Many questions remain open, such as a characterization of spy-good graphs.
In all known spy-good graphs, the spies can ensure that at the end of each
round the number of spies at any vertex $v$ is at least $\FL{r(v)/m}$, where
$r(v)$ is the number of \revs\ at $v$.  Existence of such a strategy is
preserved when vertices expand into a complete subgraph.  Also, Howard and
Smyth~\cite{HS} observed that $\sgmr$ is preserved by taking the distance power
of a graph.  Hence every graph obtained from some webbed tree via some sequence
of distance powers or vertex expansions is spy-good, but these are not the only
spy-good graphs.

It would also be interesting to bound $\sgmr$ in terms of other graph 
parameters, such as treewidth.  Generalizations of the game are also possible,
such as by allowing players to travel farther in a move or by requiring more
spies to guard a meeting.  One can also consider analogous games on directed
graphs.

\section{Dominating Vertices and Webbed Trees}\label{spygoodsec}
We begin with graphs having a {\it dominating vertex} (a vertex adjacent to all
others); we then apply this result to webbed trees.  Let $N(v)$ denote the
neighborhood of a vertex $v$.  Also $N[v]=N(v)\cup \{v\}$, and
$N(S)=\bigcup_{v\in S}N(v)$.

\begin{defin}\label{dstable}
For a graph $G$ having a dominating vertex $u$, a position in the game $\RSG$
is {\it stable} if, for each vertex $v$ other than $u$, the number of spies at
$v$ is exactly $\FL{r(v)/m}$, where $r(v)$ is the number of \revs\ at $v$.  The
other spies, if any, are at $u$.
\end{defin}

\begin{thm}\label{dom}
If a graph $G$ has a dominating vertex, then $\sgmr=\FL{r/m}$.
\end{thm}

\begin{proof}
Let $u$ be a dominating vertex in $G$, and let $s = \FL{r/m}$.  Since
$s=\FL{r/m}$, a stable position will have a spy at $u$ if there is a meeting at
$u$.  Hence a stable position has no unguarded meeting.  When $s=\FL{r/m}$,
there are enough spies to establish a stable position after the initial round.
We show that the spies can reestablish a stable position at the end of each
round.

Consider a stable position at the start of round $t$.  Let $X$ be a maximal
family of disjoint sets of $m$ \revs\ such that each set is located at one
vertex other than $u$.  Let $Y$ be such a maximal family after the \revs\ move
in round $t$.  In $X$ or $Y$, more than one set may be located at a single
vertex in $G$.  For example, a vertex $v$ having $pm+q$ \revs\ at the start of
round $t$ (where $0\le q<m$) corresponds to $p$ elements of $X$, and there are
$p$ spies at $v$ at that time.

Let $X =\{\VEC x1k\}$ and $Y = \{\VEC y1{k'}\}$.  Let $X'=\{\VEC x{k+1}s\}$,
representing the excess spies waiting at $u$ after round $t$.
Define an auxiliary bipartite graph $H$ with partite sets $X\cup X'$ and $Y$.
For $x_i \in X$ and $y_j \in Y$, put $x_iy_j\in E(H)$ if some \rev\ from
meeting $x_i$ is in meeting $y_j$ (note that $x_i$ and $y_j$ may be the same
set).  Also make all of $X'$ adjacent to all of $Y$.  If some matching in $H$
covers $Y$, then the spies can move so that every vertex other than $u$ having
$p'm+q'$ \revs\ at the end of round $t$ (where $0\le q'<m$) has exactly $p'$
spies on it (and the remaining spies are at $u$).

The existence of such a matching follows from Hall's Theorem.  For $S\esub Y$, 
always $X' \subseteq N(S)$, so $|N(S)| = |X'| + |N(S) \cap X|$.  Consider the
$m|S|$ \revs\ in the meetings corresponding to $S$.  Such \revs\ came from
meetings in $|N(S)\cap X|$ or were not in any of the $k$ meetings indexed by
$X$.  Hence $m|S|\le m|N(S)\cap X|+(r-km)$.  Since $|X'|=s-k$ and $s=\FL{r/m}$,
$$
 |N(S)| \geq |X'| + |S| - \left( \lfloor {r/m} \rfloor -k\right) 
        = s-k + |S| - \left( \lfloor {r/m} \rfloor -k\right) = |S|,
$$
so Hall's Condition holds.
\end{proof}

\begin{cor}
Fix $n,m,r$ with $n\ge r/m$.  For $0 \leq k \leq {n \choose 2}$, there is an
$n$-vertex graph $G$ with $k$ edges such that $\sgmr=\FL{r/m}$.
\end{cor}
\begin{proof}
For $k\ge n$, form $G$ by adding the desired number of edges joining leaves of
an $n$-vertex star; Theorem~\ref{dom} applies.  For $k \le n-1$, let $G$ be a
star plus isolated vertices; use Theorem~\ref{dom} and
$\FL{a}+\FL{b}\le\FL{a+b}$.
\end{proof}

\begin{defin}\label{webdef}
For any vertex $v$ in a rooted tree, the \emph{parent} of a non-root vertex
$v$ (written $v^+$) is the first vertex after $v$ on the path from $v$ to the
root.  The set of \emph{children} of $v$ (written $C(v)$) is the set of
neighbors of $v$ other than its parent, and the set of \emph{descendants} of
$v$ (written $D(v)$) is the set of vertices whose path to the root contains $v$.
A \emph{webbed tree} is a graph $G$ having a rooted spanning tree $T$ such that
every edge of $G$ outside $T$ joins two vertices having the same parent
(called \emph{siblings}).  Figure~\ref{webtree} shows a webbed tree, with the
rooted spanning tree $T$ in bold.
\end{defin}

Trivially, every tree is a webbed tree, as is every graph having a dominating
vertex.  In fact, a 2-connected graph is a webbed tree if and only if it has a
dominating vertex.  Every webbed tree is a graph whose blocks have dominating
vertices, but the converse does not hold.  Consider the graph obtained from two
$4$-cycles with a common vertex by adding chords of the $4$-cycles to create
four vertices of degree $3$; every block has a dominating vertex, but the graph
is not a webbed tree.

Our main result in this section is that all webbed trees are spy-good.  This
conclusion is proved for trees in~\cite{CSW}.  In that paper, an invariant
defined in terms of the positions of the \revs\ specifies how many spies should
be placed on each vertex.  The invariant guarantees that all meetings are
covered, and a direct proof is given to show that the spies can restore the
invariant after each round.

Here we use the same invariant to generalize the tree result to the class of
webbed trees.  Our method of proving that the invariant has the desired
properties is different from that in~\cite{CSW}.  Here we decompose the spies'
response into independent responses in imagined games on subgraphs having a
dominating vertex.  After the \revs\ move, the spies restore the invariant by
applying the strategy in Theorem~\ref{dom} independently to each graph induced
by a vertex and its children in the spanning tree.  Because we will apply
Theorem~\ref{dom}, we don't use ``stable'' for positions satisfying the
invariant in a webbed tree; instead, we reserve that term for positions in the
auxiliary local games, whose graphs have dominating vertices.

In~\cite{CSW}, the result on trees is extended in a different direction to
determine the winner in $\RSG$ whenever $G$ has at most one cycle.  A similar
extension is possible here for graphs obtained by adding a cycle through the
roots of disjoint webbed trees, but the resulting family is not as natural as
the family of unicyclic graphs.

\begin{thm}\label{webbed}
If $G$ is a webbed tree, then $\sgmr=\FL{r/m}$.
\end{thm}
\begin{proof}
Let $T$ be a rooted spanning tree in $G$ such that every edge of $G$ not in
$T$ joins sibling vertices in $T$.  Let $z$ be the root of $T$, and let
$s=\FL{r/m}$.  The notation for children and descendants is as in
Definition~\ref{webdef} with respect to $T$.

For each vertex $v$, let $r(v)$ and $s(v)$ denote the number of \revs\ and
spies on $v$ at the current time, respectively, and let $w(v)=\SM u{D(v)}r(u)$.
The spies maintain the following invariant specifying the number of spies on
each vertex at the end of any round:
\begin{equation}\label{invar}
s(v)=\floor{\frac{w(v)}{m}}-\sum_{x\in C(v)} \floor{\frac{w(x)}{m}}
\qquad\textrm{for }v\in V(G).
\end{equation}
Since $\SM x{C(v)}w(x)=w(v)-r(v)$, the formula is always nonnegative.  Also, if
$r(v)\ge m$, then $s(v)\ge \FL{\FR{w(v)}m}-\FL{\FR{w(v)-r(v)}m}\ge 1$.  Hence
(\ref{invar}) guarantees that every meeting is guarded.

To show that the spies can establish (\ref{invar}) after the first round, it
suffices that all the formulas sum to $\floor{r/m}$.  More generally, summing
over the descendants of any vertex $v$,
\begin{equation}\label{telsum}
\SM u{D(v)} s(u)= \FL{\FR{w(v)}m},
\end{equation}
since $\FL{w(u)/m}$ occurs positively in the term for $u$ and negatively in the
term for $u^+$, except that $\FL{w(v)/m}$ occurs only positively.  When $v=z$,
the total is $\FL{r/m}$, since $w(z)=r$.

To show that the spies can maintain (\ref{invar}), let $r(v)$ and $s(v)$ refer
to the start of round $t$, let $r'(v)$ denote the number of \revs\ at $v$ after
the \revs\ move in round $t$, and let $w'(v)=\SM u{D(v)}r'(v)$.  The spies will
move in round $t$ to achieve the new values required by (\ref{invar}).
To determine these moves, we will use Theorem~\ref{dom} to obtain a stable
position in each subgraph induced by a vertex and its children, independently.
Let $G(v)$ denote the subgraph induced by $C(v)\cup\{v\}$; note that $v$ is a
dominating vertex in $G(v)$.  We will play a round in an imagined ``local''
game on $G(v)$ for each vertex $v$.

\begin{figure}[hbt]
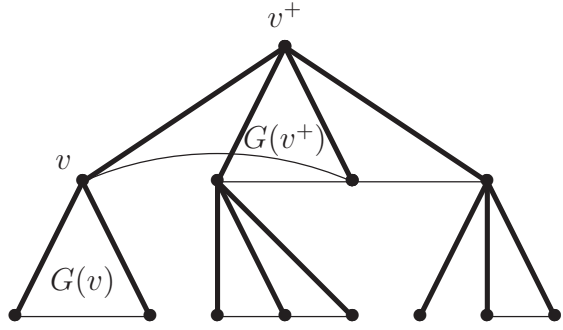

\gpic{
\expandafter\ifx\csname graph\endcsname\relax \csname newbox\endcsname\graph\fi
\expandafter\ifx\csname graphtemp\endcsname\relax \csname newdimen\endcsname\graphtemp\fi
\setbox\graph=\vtop{\vskip 0pt\hbox{%
    \graphtemp=.5ex\advance\graphtemp by 0.141in
    \rlap{\kern 1.500in\lower\graphtemp\hbox to 0pt{\hss $\bu$\hss}}%
    \graphtemp=.5ex\advance\graphtemp by 0.847in
    \rlap{\kern 0.441in\lower\graphtemp\hbox to 0pt{\hss $\bu$\hss}}%
    \graphtemp=.5ex\advance\graphtemp by 0.847in
    \rlap{\kern 1.147in\lower\graphtemp\hbox to 0pt{\hss $\bu$\hss}}%
    \graphtemp=.5ex\advance\graphtemp by 0.847in
    \rlap{\kern 1.853in\lower\graphtemp\hbox to 0pt{\hss $\bu$\hss}}%
    \graphtemp=.5ex\advance\graphtemp by 1.553in
    \rlap{\kern 0.088in\lower\graphtemp\hbox to 0pt{\hss $\bu$\hss}}%
    \graphtemp=.5ex\advance\graphtemp by 1.553in
    \rlap{\kern 0.794in\lower\graphtemp\hbox to 0pt{\hss $\bu$\hss}}%
    \graphtemp=.5ex\advance\graphtemp by 1.553in
    \rlap{\kern 1.147in\lower\graphtemp\hbox to 0pt{\hss $\bu$\hss}}%
    \graphtemp=.5ex\advance\graphtemp by 1.553in
    \rlap{\kern 1.500in\lower\graphtemp\hbox to 0pt{\hss $\bu$\hss}}%
    \graphtemp=.5ex\advance\graphtemp by 1.553in
    \rlap{\kern 1.853in\lower\graphtemp\hbox to 0pt{\hss $\bu$\hss}}%
    \graphtemp=.5ex\advance\graphtemp by 0.000in
    \rlap{\kern 1.500in\lower\graphtemp\hbox to 0pt{\hss $v^+$\hss}}%
    \graphtemp=.5ex\advance\graphtemp by 0.747in
    \rlap{\kern 0.341in\lower\graphtemp\hbox to 0pt{\hss $v$\hss}}%
    \graphtemp=.5ex\advance\graphtemp by 0.635in
    \rlap{\kern 1.500in\lower\graphtemp\hbox to 0pt{\hss $G(v^+)$\hss}}%
    \graphtemp=.5ex\advance\graphtemp by 1.376in
    \rlap{\kern 0.441in\lower\graphtemp\hbox to 0pt{\hss $G(v)$\hss}}%
    \special{pn 8}%
    \special{pa 1500 141}%
    \special{pa 1147 847}%
    \special{pa 1853 847}%
    \special{pa 1500 141}%
    \special{pa 441 847}%
    \special{pa 88 1553}%
    \special{pa 794 1553}%
    \special{pa 441 847}%
    \special{fp}%
    \special{pa 1147 847}%
    \special{pa 1147 1553}%
    \special{pa 1500 1553}%
    \special{pa 1853 1553}%
    \special{pa 1147 847}%
    \special{pa 1500 1553}%
    \special{fp}%
    \special{ar 1147 2464 1765 1765 -1.982313 -1.159279}%
    \graphtemp=.5ex\advance\graphtemp by 0.141in
    \rlap{\kern 1.500in\lower\graphtemp\hbox to 0pt{\hss $\bu$\hss}}%
    \graphtemp=.5ex\advance\graphtemp by 0.847in
    \rlap{\kern 1.853in\lower\graphtemp\hbox to 0pt{\hss $\bu$\hss}}%
    \graphtemp=.5ex\advance\graphtemp by 0.847in
    \rlap{\kern 2.559in\lower\graphtemp\hbox to 0pt{\hss $\bu$\hss}}%
    \graphtemp=.5ex\advance\graphtemp by 1.553in
    \rlap{\kern 2.206in\lower\graphtemp\hbox to 0pt{\hss $\bu$\hss}}%
    \graphtemp=.5ex\advance\graphtemp by 1.553in
    \rlap{\kern 2.559in\lower\graphtemp\hbox to 0pt{\hss $\bu$\hss}}%
    \graphtemp=.5ex\advance\graphtemp by 1.553in
    \rlap{\kern 2.912in\lower\graphtemp\hbox to 0pt{\hss $\bu$\hss}}%
    \special{pa 2206 1553}%
    \special{pa 2559 847}%
    \special{pa 2559 1553}%
    \special{pa 2912 1553}%
    \special{pa 2559 847}%
    \special{pa 1500 141}%
    \special{pa 1853 847}%
    \special{pa 2559 847}%
    \special{fp}%
    \special{pn 28}%
    \special{pa 1500 141}%
    \special{pa 2559 847}%
    \special{fp}%
    \special{pa 1500 141}%
    \special{pa 1853 847}%
    \special{fp}%
    \special{pa 1500 141}%
    \special{pa 441 847}%
    \special{fp}%
    \special{pa 1500 141}%
    \special{pa 1147 847}%
    \special{fp}%
    \special{pa 88 1553}%
    \special{pa 441 847}%
    \special{fp}%
    \special{pa 441 847}%
    \special{pa 794 1553}%
    \special{fp}%
    \special{pa 1147 847}%
    \special{pa 1853 1553}%
    \special{fp}%
    \special{pa 1147 847}%
    \special{pa 1147 1553}%
    \special{fp}%
    \special{pa 1147 847}%
    \special{pa 1500 1553}%
    \special{fp}%
    \special{pa 2559 847}%
    \special{pa 2912 1553}%
    \special{fp}%
    \special{pa 2559 847}%
    \special{pa 2206 1553}%
    \special{fp}%
    \special{pa 2559 847}%
    \special{pa 2559 1553}%
    \special{fp}%
    \graphtemp=.5ex\advance\graphtemp by 0.141in
    \rlap{\kern 1.500in\lower\graphtemp\hbox to 0pt{\hss $\bu$\hss}}%
    \graphtemp=.5ex\advance\graphtemp by 0.847in
    \rlap{\kern 0.441in\lower\graphtemp\hbox to 0pt{\hss $\bu$\hss}}%
    \graphtemp=.5ex\advance\graphtemp by 0.847in
    \rlap{\kern 1.147in\lower\graphtemp\hbox to 0pt{\hss $\bu$\hss}}%
    \graphtemp=.5ex\advance\graphtemp by 0.847in
    \rlap{\kern 1.853in\lower\graphtemp\hbox to 0pt{\hss $\bu$\hss}}%
    \graphtemp=.5ex\advance\graphtemp by 1.553in
    \rlap{\kern 0.088in\lower\graphtemp\hbox to 0pt{\hss $\bu$\hss}}%
    \graphtemp=.5ex\advance\graphtemp by 1.553in
    \rlap{\kern 0.794in\lower\graphtemp\hbox to 0pt{\hss $\bu$\hss}}%
    \graphtemp=.5ex\advance\graphtemp by 1.553in
    \rlap{\kern 1.147in\lower\graphtemp\hbox to 0pt{\hss $\bu$\hss}}%
    \graphtemp=.5ex\advance\graphtemp by 1.553in
    \rlap{\kern 1.500in\lower\graphtemp\hbox to 0pt{\hss $\bu$\hss}}%
    \graphtemp=.5ex\advance\graphtemp by 1.553in
    \rlap{\kern 1.853in\lower\graphtemp\hbox to 0pt{\hss $\bu$\hss}}%
    \graphtemp=.5ex\advance\graphtemp by 0.141in
    \rlap{\kern 1.500in\lower\graphtemp\hbox to 0pt{\hss $\bu$\hss}}%
    \graphtemp=.5ex\advance\graphtemp by 0.847in
    \rlap{\kern 1.853in\lower\graphtemp\hbox to 0pt{\hss $\bu$\hss}}%
    \graphtemp=.5ex\advance\graphtemp by 0.847in
    \rlap{\kern 2.559in\lower\graphtemp\hbox to 0pt{\hss $\bu$\hss}}%
    \graphtemp=.5ex\advance\graphtemp by 1.553in
    \rlap{\kern 2.206in\lower\graphtemp\hbox to 0pt{\hss $\bu$\hss}}%
    \graphtemp=.5ex\advance\graphtemp by 1.553in
    \rlap{\kern 2.559in\lower\graphtemp\hbox to 0pt{\hss $\bu$\hss}}%
    \graphtemp=.5ex\advance\graphtemp by 1.553in
    \rlap{\kern 2.912in\lower\graphtemp\hbox to 0pt{\hss $\bu$\hss}}%
    \hbox{\vrule depth1.641in width0pt height 0pt}%
    \kern 3.000in
  }%
}%
}
\caption{Decomposition of a webbed tree}\label{webtree}
\end{figure}

To set up the local games, we partition the $s(v)$ spies at each vertex $v$
into a set of $\ds(v)$ spies to be used in the local game on $G(v)$ and a set
of $\us(v)$ spies to be used in the local game on $G(v^+)$, where $\ds(v)$ and
$\us(v)$ sum to $s(v)$ (when the tree is drawn with the root $z$ at the top,
the accent indicates the direction of the relevant subgraph).

Let $D^*(v)=D(v)-\{v\}$.  Let $w^*(v)$ be the number of \revs\ that are in
$D^*(v)$ at the start of round $t$ {\it or} are there after the \revs\ move in
round $t$.  Every \rev\ counted by $w^*(v)$ is also counted by $w(v)$, and
every \rev\ counted by $\SM x{C(v)}w(x)$ is also counted by $w^*(v)$.  These
statements also hold with $w'$ in place of $w$.  Hence
\begin{equation}\label{nonneg}
w(v)\ge w^*(v)
\qquad\textrm{and}\qquad
w^*(v)\ge \SM x{C(v)}w(x).
\end{equation}
By (\ref{nonneg}),  $\us(v)$ and $\ds(v)$ are nonnegative when we define
\begin{equation}\label{supdown}
\us(v)=\FL{\FR{w(v)}m}-\FL{\FR{w^*(v)}m}
\qquad\textrm{and}\qquad
\ds(v)=\FL{\FR{w^*(v)}m}-\SM x{C(v)}\FL{\FR{w(x)}m}.
\end{equation}
By (\ref{invar}), $\us(v)+\ds(v)=s(v)$.  Note also that if $v$ is a leaf of
$T$, then $\ds(v)=0$ and $\us(v)=s(v)$.

For each non-leaf vertex $v$, the spies first imagine positions of \revs\ in a
game on the graph $G(v)$ that together with (\ref{supdown}) for the spies form
a stable position.  After viewing the actual moves by \revs\ within $G(v)$ as
moves in this game, the spies reestablish stability as in Theorem~\ref{dom}.
We will show that the resulting positions satisfy the global invariant.  The
spies imagine $\ur(v)$ spies at $v$ in $G(v^+)$ and $\dr(v)$ spies at $v$ in
$G(v)$, where

\begin{equation}\label{rupdown}
\ur(v)=w(v)-m\FL{\FR{w^*(v)}m}
\qquad\textrm{and}\qquad
\dr(v)=w^*(v)-\SM x{C(v)}{w(x)}.
\end{equation}
By (\ref{nonneg}), the values of $\dr(v)$ and $\ur(v)$ are nonnegative.
Furthermore, we claim that if (\ref{supdown}) and (\ref{rupdown}) hold at each
vertex $v$, then the position on each subgraph induced by one parent and its
children is stable.  In $G(v)$ we use $\ds(v)$ and $\dr(v)$, and we use $\us(x)$
and $\ur(x)$ for $x\in C(v)$.  By definition, $\us(x)=\FL{\ur(x)/m}$.  It
remains only to check the sum.  We compute the total number of \revs\ in the
local game:
$$
\dr(v)+\SM x{C(v)}\ur(x) =
w^*(v)-\SM x{C(v)}{w(x)}+ \SM x{C(v)} w(x)-m\SM x{C(v)}\FL{\FR{w^*(x)}m}
$$
Dividing by $m$ yields $\FR{w^*(v)}m-\SM x{C(v)}\FL{\FR{w^*(x)}m}$, whose floor
is $\ds(v)+\SM x{C(v)}\us(x)$, as desired.

The spies next view the actual moves by \revs\ in the global game as moves
by the \revs\ in the imagined local games.  Each such move occurs within the
subgraph $G(v)$ for one vertex $v$.  The local game can model these moves if
the relevant value of $\ur$ or $\dr$ is at least the number of real \revs\
leaving this vertex and staying within this subgraph.  The \revs\ leaving $v$
by edges in $G(v^+)$ are those that were in $D(v)$ and now are not; there are
at most $w(v)-w^*(v)$ of them.  By (\ref{rupdown}), $\ur(v)$ is
at least this large.  Similarly, \revs\ leaving $v$ via $G(v)$ wind up in
$D^*(v)$ but were not there previously, so the number of them is at most
$w^*(v)-\SM x{C(v)} w(x)$, which equals $\dr(v)$.

The net change in the actual number of \revs\ at $v$ is $r'(v)-r(v)$.  Some of
this change is due to moves in $G(v)$ and the rest to moves in $G(v^+)$.  Moves
in $G(v^+)$ enter or leave $D(v)$.  Hence the net change in the number of
\revs\ at $v$ due to such moves is $w'(v)-w(v)$.  The remaining net change,
due to moves between $v$ and its children (in $G(v)$), is
$(r'(v)-r(v))-(w'(v)-w(v))$.  Therefore, after executing the actual moves in
the imagined local games, the new imagined distributions for the \revs\ are
given by
\begin{equation}\label{rchange}
\ur'(v)=\ur(v)+w'(v)-w(v)
\quad\textrm{and}\quad
\dr'(v)=\dr(v)+(r'(v)-r(v))-(w'(v)-w(v)).
\end{equation}
The specification of $\ur(v)$ in (\ref{rupdown}) and the change from $\ur(v)$
to $\ur'(v)$ in (\ref{rchange}) immediately yield the formula for $\ur'(v)$
in (\ref{rnew}).  To obtain $\dr'(v)$, start with the formula for $\dr'(v)$ in
(\ref{rupdown}) and adjust by the definitions of $r(v)-r(v)$ and $w'(v)-r'(v)$,
as indicated in (\ref{rchange}).  We compute
\begin{align*}
\dr'(v)&=\dr(v)+(w(v)-r(v))-(w'(v)-r'(v))\\
&=w^*(v)-\SM x{C(v)}w(x)+\SM x{C(v)}w(x)-\SM x{C(v)}w'(x)
=w^*(v)-\SM x{C(v)}{w'(x)}.
\end{align*}
Thus
\begin{equation}\label{rnew}
\ur'(v)=w'(v)-m\FL{\FR{w^*(v)}m}
\qquad\textrm{and}\qquad
\dr'(v)=w^*(v)-\SM x{C(v)}{w'(x)}.
\end{equation}

The spies now respond in the local games.  By Theorem~\ref{dom}, these
positions are stable, so $\us'(x)=\FL{\ur'(x)/m}$ for $x\in C(v)$, and
$\ds'(v)$ is the leftover amount for $v$ in the local game on $G(v)$.  By the
same computation that earlier showed $\ds(v)$ was the correct needed amount of
spies left for $v$ in $G(v)$, also
$\ds'(v)= \FL{\FR{w^*(v)}m}-\SM x{C(v)}\FL{\FR{w'(x)}m}$.

Because each spy participated in exactly one local game, playing the local
games independently ensures automatically that each spy moves at most once
in round $t$.  Hence the spy moves we have described are feasible.  It remains
only to show that (\ref{invar}) holds for the resulting distribution of spies;
that is
$$
\us'(v)+\ds'(v)=\floor{\frac{w'(v)}{m}}-\sum_{x\in C(v)} \floor{\frac{w'(x)}{m}}
\qquad\textrm{for }v\in V(G).
$$
Since the terms involving $w^*$ again cancel, we use (\ref{rnew}) to show that
$\us'(v)+\ds'(v)$ equals the desired value $s'(v)$ in the same way we used
(\ref{rupdown}) to show that the invented values $\us(v)$ and $\ds(v)$ sum to
$s(v)$.
\end{proof}

\section{Spy-good vs.~Spy-bad}\label{goodbadsec}

It is not true that all spy-good graphs are webbed trees.  Given $G$, let $G^k$
denote the graph defined by $V(G^k)=V(G)$ and $E(G^k)=\{uv\st d_G(u,v)\le k\}$.
The spies can simulate one round of the game on $G^k$ by playing $k$ rounds on
$G$.  Thus $\sigma(G^k,m,r)\le \sgmr$, as noted by Howard and Smyth~\cite{HS}.
This makes the square of a webbed tree spy-good, even though it is not
generally a webbed tree (consider $G=P_n$, for example).

Say that a spy strategy is {\it conformal} if at the end of each round the
number of spies at each vertex $v$ is at least $\FL{r(v)/m}$, where $r(v)$ is
the number of \revs\ there.  For any conformal spy strategy on $G$, the strategy
described above for $G^k$ is also conformal.  Another graph operation also
preserves the existence of conformal strategies.

\begin{prop}
Obtain $G'$ from a graph $G$ by expanding a vertex of $G$ into a clique.  If
$\FL{r/m}$ spies win $\RSG$ by a conformal strategy, then the same holds for
$G'$.
\end{prop}
\begin{proof}
Let $Q$ be the clique into which vertex $v$ of $G$ is expanded to form $G'$.
The spies play on $G'$ by imagining a game on $G$.  At each round, the \revs\
on $Q$ in $G'$ are collected onto $v$ in $G$, with $r(v)$ there after the
previous round and $r'(v)$ after the \revs\ move.  For other vertices, the
amounts before and after are as in the real game on $G'$.

Since $\sum\FL{a_i}\le \FL{\sum a_i}$, the spies on $v$ at the end of the round
in $G$ suffice to cover the $r'(v)$ \revs\ on $Q$ in $G$ and can move there,
since all vertices of $Q$ have the same neighbors outside $Q$ that $v$ has in
$G$.  Extra spies move to any vertex of $Q$.  Movements of spies from $v$ in
$G$ can also be matched by moves in the game on $G'$.  Other movements are the
same in $G$ and $G'$.  This produces a conformal strategy on $G'$.
\end{proof}

\begin{prop}
On a webbed tree $G$, the winning strategy in Theorem~\ref{webbed} is conformal.
\end{prop}
\begin{proof}
Let $T$ be a rooted spanning tree such that edges outside $T$ join siblings in
$T$.  After each round, the number of spies on vertex $v$ is given by
$$
\FL{\FR{r(v)+\SM x{C(v)} w(x)}m}-\SM x{C(v)} \FL{\FR{w(x)}m}.
$$
Since $\sum\FL{a_i}\le \FL{\sum a_i}$, the strategy is conformal.
\end{proof}

These results imply that graphs obtained from webbed trees by vertex
expansions and distance powers are spy-good.  For example, the square of
a path is spy-good.  This graph is not a webbed tree, since it is 2-connected
but has no dominating vertex (when it has at least six vertices).  On the
other hand, it is an interval graph, where an {\it interval graph} is a graph
representable by assigning each vertex $v$ an interval on the real line so that
vertices are adjacent if and only if their intervals intersect.  An interval
graph that is not a distance power and has no two vertices with the same
closed neighborhood is obtained from the square of an $8$-vertex path by adding
an edge joining the third and sixth vertices.

\begin{ques}
Which graphs are spy-good?
\end{ques}

We believe that all interval graphs are spy-good, even though
the class is not contained in the spy-good classes obtained above.

Although not all graphs are spy-good, Theorem~\ref{dom} yields good upper
bounds on $\sgmr$ for graphs with small dominating sets.
A {\it dominating set} in a graph $G$ is a set $S\esub V(G)$ such that every
vertex outside $S$ has a neighbor in $S$; the {\it domination number}
$\gamma(G)$ is the minimum size of a dominating set in $G$.

\begin{cor}\label{domnum}
$\sgmr\le\gamma(G)\FL{r/m}$ for any graph $G$.
\end{cor}
\begin{proof}
Let $S$ be a smallest dominating set.  With each vertex $u\in S$, associate
$\FL{r/m}$ spies.  Let $G_u$ be the subgraph of $G$ induced by $N[u]$; it has
$u$ as a dominating vertex.  The spies associated with $u$ stay in $G_u$,
following the strategy of Theorem~\ref{dom} on $G_u$.  When there are fewer
than $r$ \revs\ in $G_u$, the spies imagine that the missing ones are at $u$.
When a real \rev\ comes to vertex $v$ in $G_u$ from outside $G_u$, a \rev\ in
the imagined game moves from $u$ to $v$ to perform its moves.  When the real
\rev\ leaves $G_u$, the \rev\ tracking it in the game on $G_u$ returns to $u$.
These moves are possible, since $u$ is a dominating vertex in $G_u$.  Since the
spies win each imagined game, the \revs\ in the real game never make an
unguarded meeting at the end of a round.
\end{proof}

As remarked in the introduction, Corollary~\ref{domnum} is of interest only
when $\gamma(G)\le m$, because otherwise the trivial upper bound $r-m+1$
is stronger.  When $\gamma(G)\le m$, the bound in Corollary~\ref{domnum} cannot
be improved.  To motivate the proof, we first present a simple construction
of spy-bad graphs.

A {\it split graph} is a graph whose vertices can be partitioned into a clique
and an independent set.  A {\it chordal graph} is a graph in which every cycle
of length at least $4$ has a chord; split graphs clearly have this property.
Recall that for fixed $r$ and $m$ a graph is spy-bad if the revolutionaries can
beat $r-m$ spies ($r-m+1$ spies trivially win).

\begin{prop}\label{split}
Given $r,m\in\NN$, there is a chordal graph $G$ (in fact a split graph)
such that $\sgmr=r-m+1$.
\end{prop}
\begin{proof}
Let $G_{m,r}$ be the split graph consisting of a clique $Q$ of size $r$ and an
independent set $S$ of size $\CH rm$, with the neighborhoods of the vertices in
$S$ being distinct $m$-sets in $Q$.  We show that $r-m$ spies cannot win.

The \revs\ initially occupy each vertex of $Q$.  Let $s'$ be the number of
vertices of $Q$ initially occupied by spies.  The number of threatened
meetings that spies on $Q$ are not adjacent to is $\CH{r-s'}m$.  Protecting
against such threats requires putting spies initially on the $\CH{r-s'}m$
vertices of $S$ corresponding to these $m$-sets, but only $r-m-s'$ remaining
spies are available, and $\CH{r-s'}m>r-m-s'$ when $r-s'\ge m$.
\end{proof}

Note that $\FR{r-m+1}{r/m}$ can be made arbitrarily large.  When $r=2m$, the
ratio exceeds $m/2$.  Letting $m$ also grow, we observe that $\sgmr$ cannot
be bounded by a constant multiple of $r/m$, even on split graphs.  Furthermore,
the strategy for \revs\ in Proposition~\ref{split} does not use any edges
within the clique, so the statement remains true also for the bipartite graph
obtained by deleting those edges.

When $m$ grows, the degrees of all vertices in $G_{m,r}$ also grow.  If the
degrees in the independent set are bounded, then the spies can do better.  We
state the next result without proof, because the proof is a bit technical and
the class of graphs is somewhat specialized.  The technique is as usual for
upper bounds: defining stable positions and showing that the spies can
reestablish a stable position after each round.  The proof will appear in the
thesis of the third author.

\begin{thm}
Let $G$ be a split graph with clique $Q$ and independent set $S$ in which each
vertex of $S$ has degree at most $d$.  If $m$ is a multiple of $d$, then 
$\sgmr\le d\CL{r/m}$.
\end{thm}

A construction like that of Proposition~\ref{split} enables us to show that
Corollary~\ref{domnum} is nearly sharp.  When $t=m$, the upper and lower bounds
in this result are equal; when $m\mid r$, the difference between them is $t-1$.

\begin{thm}\label{domsharp}
Given $t,m,r\in\NN$ such that $t\le m\le r-m$, there is a graph $G$ with
domination number $t$ such that $\sgmr>t(r/m-1)$.
\end{thm}
\begin{proof}
First we construct a graph $G$.  Begin with a copy of $K_{t,r}$ having partite
sets $T$ of size $t$ and $R$ of size $r$.  Add an independent set $U$ of size
$t\CH rm$, grouped into sets of size $t$.  With each $m$-set $A$ in $R$,
associate one $t$-set $A'$ in $U$.  Make all of $A$ adjacent to all of $A'$,
and add a matching joining $A'$ to $T$ (see Figure~\ref{domshpfig}).  Note that
$T$ is a dominating set.

\begin{figure}[hbt]
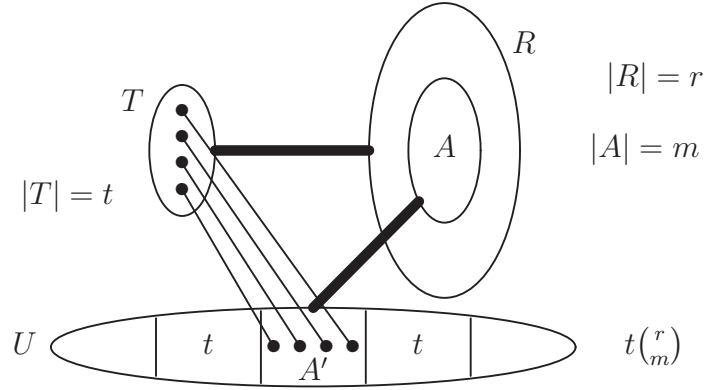

\gpic{
\expandafter\ifx\csname graph\endcsname\relax \csname newbox\endcsname\graph\fi
\expandafter\ifx\csname graphtemp\endcsname\relax \csname newdimen\endcsname\graphtemp\fi
\setbox\graph=\vtop{\vskip 0pt\hbox{%
    \special{pn 11}%
    \special{ar 824 772 172 343 0 6.28319}%
    \special{ar 2197 772 189 378 0 6.28319}%
    \special{ar 2197 772 395 772 0 6.28319}%
    \special{ar 1510 1802 1373 206 0 6.28319}%
    \special{pn 56}%
    \special{pa 996 772}%
    \special{pa 1802 772}%
    \special{fp}%
    \special{pa 2063 1039}%
    \special{pa 1510 1596}%
    \special{fp}%
    \graphtemp=.5ex\advance\graphtemp by 0.530in
    \rlap{\kern 0.565in\lower\graphtemp\hbox to 0pt{\hss $T$\hss}}%
    \graphtemp=.5ex\advance\graphtemp by 0.772in
    \rlap{\kern 2.197in\lower\graphtemp\hbox to 0pt{\hss $A$\hss}}%
    \graphtemp=.5ex\advance\graphtemp by 0.226in
    \rlap{\kern 2.613in\lower\graphtemp\hbox to 0pt{\hss $R$\hss}}%
    \graphtemp=.5ex\advance\graphtemp by 1.015in
    \rlap{\kern 0.222in\lower\graphtemp\hbox to 0pt{\hss ${|T|=t}$\hss}}%
    \graphtemp=.5ex\advance\graphtemp by 0.398in
    \rlap{\kern 3.300in\lower\graphtemp\hbox to 0pt{\hss ${|R|=r}$\hss}}%
    \graphtemp=.5ex\advance\graphtemp by 1.802in
    \rlap{\kern 0.000in\lower\graphtemp\hbox to 0pt{\hss $ U$\hss}}%
    \graphtemp=.5ex\advance\graphtemp by 1.802in
    \rlap{\kern 3.295in\lower\graphtemp\hbox to 0pt{\hss $t\CH rm$\hss}}%
    \special{pn 11}%
    \special{pa 687 1651}%
    \special{pa 687 1953}%
    \special{fp}%
    \special{pa 1236 1613}%
    \special{pa 1236 1991}%
    \special{fp}%
    \special{pa 1785 1613}%
    \special{pa 1785 1991}%
    \special{fp}%
    \special{pa 2334 1651}%
    \special{pa 2334 1953}%
    \special{fp}%
    \graphtemp=.5ex\advance\graphtemp by 0.978in
    \rlap{\kern 0.824in\lower\graphtemp\hbox to 0pt{\hss $\bu$\hss}}%
    \graphtemp=.5ex\advance\graphtemp by 0.841in
    \rlap{\kern 0.824in\lower\graphtemp\hbox to 0pt{\hss $\bu$\hss}}%
    \graphtemp=.5ex\advance\graphtemp by 0.704in
    \rlap{\kern 0.824in\lower\graphtemp\hbox to 0pt{\hss $\bu$\hss}}%
    \graphtemp=.5ex\advance\graphtemp by 0.566in
    \rlap{\kern 0.824in\lower\graphtemp\hbox to 0pt{\hss $\bu$\hss}}%
    \graphtemp=.5ex\advance\graphtemp by 1.802in
    \rlap{\kern 1.304in\lower\graphtemp\hbox to 0pt{\hss $\bu$\hss}}%
    \graphtemp=.5ex\advance\graphtemp by 1.802in
    \rlap{\kern 1.442in\lower\graphtemp\hbox to 0pt{\hss $\bu$\hss}}%
    \graphtemp=.5ex\advance\graphtemp by 1.802in
    \rlap{\kern 1.579in\lower\graphtemp\hbox to 0pt{\hss $\bu$\hss}}%
    \graphtemp=.5ex\advance\graphtemp by 1.802in
    \rlap{\kern 1.716in\lower\graphtemp\hbox to 0pt{\hss $\bu$\hss}}%
    \special{pa 824 978}%
    \special{pa 1304 1802}%
    \special{fp}%
    \special{pa 824 841}%
    \special{pa 1442 1802}%
    \special{fp}%
    \special{pa 824 704}%
    \special{pa 1579 1802}%
    \special{fp}%
    \special{pa 824 566}%
    \special{pa 1716 1802}%
    \special{fp}%
    \graphtemp=.5ex\advance\graphtemp by 1.802in
    \rlap{\kern 0.961in\lower\graphtemp\hbox to 0pt{\hss $t$\hss}}%
    \graphtemp=.5ex\advance\graphtemp by 1.940in
    \rlap{\kern 1.510in\lower\graphtemp\hbox to 0pt{\hss ${A'}$\hss}}%
    \graphtemp=.5ex\advance\graphtemp by 1.802in
    \rlap{\kern 2.060in\lower\graphtemp\hbox to 0pt{\hss $t$\hss}}%
    \graphtemp=.5ex\advance\graphtemp by 0.772in
    \rlap{\kern 3.244in\lower\graphtemp\hbox to 0pt{\hss $|A|=m$\hss}}%
    \hbox{\vrule depth2.008in width0pt height 0pt}%
    \kern 3.300in
  }%
}%
}
\caption{Sharpness of the domination bound\label{domshpfig}}
\end{figure}

To show that $\gamma(G)=t$, let $S$ be a smallest dominating set.  For each
$m$-set $A$ in $R$, the $t$ vertices in $A'$ are adjacent only to $A$ in $R$.
Thus if $\C{S\cap R}<t\le r-m$, then some $t$-set $A'$ in $U$ is
undominated by $S\cap R$.  Outside of $R$, the closed neighborhoods of the
vertices in $A'$ are pairwise disjoint, so $S$ needs $t$ additional vertices to
dominate them.  Hence $|S|\ge t$.

Now, we give a strategy for the \revs\ to win against $t(r/m-1)$ spies on $G$.
Let $s=\FL{t(r/m-1)}$.
The \revs\ initially occupy $R$, one on each vertex.  A spy on a vertex $u$ of
$U$ can protect all the same threats (and more) by locating at the neighbor of
$u$ in $T$ instead.  Hence we may assume (at least for the purpose of trying to
survive the next round) that no spies locate initially in $U$.

Let $v$ be a vertex of $T$ having the fewest initial spies, and let $s(v)$ be
the number of spies there.  The \revs\ will win by attacking the neighbors of
$v$.  Let $s'$ be the number of spies initially in $R$, so $s(v)\le (s-s')/t$.

The \revs\ want to form meetings at $s(v)+1$ neighbors of $v$ that are
neighbors of no other vertices with spies.  Let $R'$ be the set of vertices
in $R$ that do not have spies; note that $|R'|\ge r-s'$.  If
$|R'|\ge m(s(v)+1)$, then the \revs\ win as follows.  First, group vertices in
$R'$ into $s(v)+1$ sets of size $m$.  For each such set $A$, the \revs\ on $A$
move to the unique vertex $u_{A,v}$ in the associated subset $A'$ of $U$ that
is adjacent to $v$ in $T$.  For each such vertex, the only neighbor having a
spy is $v$, so the meetings cannot all be guarded and the \revs\ win.

It thus suffices to show that $r-s'\ge m(s(v)+1)$.  Since $v$ has the fewest
spies among vertices of $T$, we have $ts(v)\le s-s'\le t(r/m-1)-s'$.
Multiplying by $m/t$ and adding $m$ yields $m(s(v)+1)\le r-s'(m/t)\le r-s'$, as
desired, using $t\le m$ at the end.
\end{proof}

Although the construction in Theorem~\ref{domsharp} depends heavily on $m$,
it does not depend much on $r$.  Indeed, the construction works equally well
whenever the number of \revs\ is at most $r$, because the \revs\ can use the
strategy for a smaller number of \revs\ on the appropriate subgraph of the 
graph constructed for $r$ \revs.  The same comment applies to
Proposition~\ref{split}.

\section{Hypercubes and Retracts}\label{hypersec}
For $d\in\NN$, let $[d]=\{1,\ldots,d\}$.  The $d$-dimensional hypercube $Q_d$
is the graph with vertex set $\{v_S\st S\esub [d]\}$ such that $v_S$ and $v_T$
are adjacent when the symmetric difference of $S$ and $T$ has size $1$.  The
{\it weight} of the vertex $v_S$ is $|S|$.  For vertices of small weight, we
write the subscripts without set brackets.  We show first that $Q_d$ is spy-bad
for $m=2$ when $d\ge r$.  For larger $m$, we will later obtain a lower bound on
$\sigma(Q_d,m,r)$ using the same basic idea.

\begin{thm}\label{hypercube}
If $G=Q_d$ and $d\ge r$, then $\sgtr=r-1$.
\end{thm}
\begin{proof}
The upper bound is trivial; we show that $r-2$ spies cannot win.
The \revs\ begin by occupying $\VEC v1r$, threatening meetings of size $2$ at
$\nul$ and at $\CH r2$ vertices of weight $2$.  Let $t$ be the number of \revs\
left uncovered by the initial placement of the spies.  Threats at $\CH t2$
vertices must be watched by spies not on vertices of weight $1$.  A spy at a
vertex of weight $2$ can watch one such threat; spies at vertices of weight $3$
can watch three of them.  Hence $s\ge (r-t)+\FR13\CH{t}2$ if the spies stop the
\revs\ from winning on the first round.  This yields $s\ge r-1$ if $t\ge5$ or
$t\le2$.

If $t=4$ and $s=r-2$, then the spies need to watch six threats at weight $2$
using two spies at vertices of weight $3$.  A spy at a vertex of weight $3$
watches the three pairs in its name.  The four uncovered \revs\ threaten 
meetings at six vertices of weight $3$ corresponding to the edges of the
complete graph $K_4$.  A spy at weight $3$ can watch three pairs corresponding
to a triangle.  Since the edges of $K_4$ cannot be covered with two triangles,
$r-2$ spies are not enough when $t=4$.

If $t=3$, then the counting bound yields $s\ge r-2$ for spies to avoid losing
on the first round.  If the initial placement of $r-2$ spies can watch all
immediate threats, then they must cover $r-3$ \revs\ at vertices of weight $1$
and occupy one vertex at weight $3$.  By symmetry, we may assume the spies
locate at $v_{123}$ and $\VEC v4r$.

In the first round, \revs\ at $v_1$ and $v_2$ move to $v_{\nul}$; the others
wait where they are.  To guard the meeting at $v_{\nul}$, a spy at some vertex
of weight $1$ must move there; let $v_j$ be the vertex from which a spy moves
to $v_{\nul}$.

In the second round, the \revs\ at $v_3$ and $v_j$ move to $v_{3j}$, winning.
The distance from each spy to $v_{3j}$ after round $1$ is at least $3$, except
for the spy at $v_{j}$, so no other spy could have moved after round 1 to
watch that threat.
\end{proof}

Extra spies on vertices of weight at least $5$ cannot prevent the \revs\ from
winning with the strategy given in the proof of Theorem~\ref{hypercube}.
This enables the \revs\ to win against somewhat fewer spies when $r$ is larger
than the dimension.

A {\it code} with length $d$ and distance $k$ is a set of vertices in $Q_d$
such that the distance between any two of them is at least $k$.  Let $A(d,k)$
denote the maximum size of a code with distance $k$ in $Q_d$, and let $B(d,k)$
be the number of vertices with distance less than $k$ from a fixed vertex in
$Q_d$.  Note that $B(d,k)=\SE i0{k-1} \CH di < d^{k-1}$ when $k>2$.
If $M<2^d/B(d,k)$, then any code of size $M$ having distance $k$ can
be extended by adding some vertex, so $A(d,k)\ge 2^d/d^{k-1}$ when $k>2$.

\begin{cor}\label{smallhyp}
If $d<r\le 2^d/d^7$, then $\sigma(Q_d,2,r)\ge (d-1)\FL{r/d}$.
\end{cor}
\begin{proof}
Let $X$ be a code in $Q_d$ with distance $9$ and size at least $2^d/d^8$.  The
\revs\ devote $d$ \revs\ to playing the strategy in the proof of
Theorem~\ref{hypercube} at each of $\FL{r/d}$ vertices of $X$.  If the ball of
radius $4$ at any such vertex has fewer than $d-1$ spies in the initial
configuration, then the \revs\ win in that ball in two rounds, since any spy
initially outside that ball is too far away to guard a meeting formed at
distance $2$ from the central point in round $2$.

Since the code has distance $9$, the balls of radius $4$ are disjoint.  Hence
$(d-1)\FL{r/d}$ spies are needed to keep the \revs\ from winning within two
rounds.
\end{proof}

Theorem~\ref{hypercube} and Corollary~\ref{smallhyp} together imply
that at least $(d-1)\FL{r/d}$ spies are needed to win against $r$ \revs\ on
$Q_d$ unless $d<\log_2 r+7\log_2\log_2 r$.  That many spies may not be enough,
since three \revs\ easily defeat one spy on $Q_2$ by starting initially at
distinct vertices.  Although four \revs\ can threaten meetings at all eight
vertices of $Q_3$, two spies can watch all those meetings and survive the next
round.  It appears that $\sigma(Q_3,2,4)=2$, though we have not worked out a
complete strategy for two spies against four \revs.  We have no nontrivial
general upper bounds on $\sigma(Q_d,2,r)$ when $r>d$.
 
Next we consider the game on hypercubes when $m>2$.  Again we use the threats
made by \revs\ placed initially at vertices of weight 1.  However, for larger
$m$ we use a probabilistic argument instead of explicit counting.  The
probabilistic arguments are simpler and yield a stronger lower bound on
$\sigma(Q_d,m,r)$ than the counting arguments would, but we no longer
completely determine the threshold (and hence we separate this from the case
$m=2$).  Again $V(Q_d)=\{v_S\st S\esub[d]\}$, as specified as before
Theorem~\ref{hypercube}.


\begin{lem}\label{lem:prefix}
For $v\in V(Q_d)$, a vertex $u$ of weight $m$ is within distance $m-1$ of
$v$ if and only if $\C{u\cap v}\ge \FR{\C v+1}2$.
\end{lem}
\begin{proof}
The distance between any two vertices is their symmetric difference.
Always the size of the symmetric difference is $\C u+\C v-2\C{u\cap v}$.  When
$\C u=m$, it follows that ${d_{Q_d}(u,v)\le m-1}$ is equivalent to
$\C{u\cap v}\ge \FR{\C v+1}2$.
\end{proof}


Our main tool for the game on $Q_d$ is a lemma about families of sets.

\begin{lem}\label{lem:prob}
Let $S$ be a set of at most $t$ vertices in $Q_t$, all having weight at least
$2$.  If $t\ge 38.73m$, then $Q_t$ has a vertex $w$ of weight $m$ such that
$d_{Q_t}(v,w)\ge m$ for all $v\in S$.
\end{lem}
\begin{proof}
Fix $p\in(0,1)$, to be determined later.  Construct a random index set
$I\esub[t]$ by independently including each element of $[t]$ with probability
$p$.  In light of Lemma~\ref{lem:prefix}, for $v\in S$ we say that $I$
\emph{avoids} $v$ if $|v\cap I| < \frac{\sizeof{v} + 1}{2}$.  Our goal is to
show that with $p$ chosen appropriately, with positive probability $I$ avoids
all of $S$ and has size at least $m$.  The desired vertex $w$ can then be any
vertex of weight $m$ contained in such a set $I$.  Our first task is to obtain
a lower bound on $\pr[A_v]$, where $A_v$ is the event that $I$ avoids $v$.  

Let $\bin(n,p)$ denote a random variable having the binomial distribution with
$n$ trials and success probability $p$.  Let $B$ be the event that $2k+1$
trials yield $k$ successes in the first $2k-1$ trials plus two failures at the
end.  Let $B'$ be the event that $2k+1$ trials yield $k-1$ successes in the
first $2k-1$ trials plus two successes at the end.  Canceling common factors
yields $\pr[B]>\pr[B']$ if and only if $p<1/2$.  As a consequence,
$\pr[\bin(2k+1,p)<k+1]>\pr[\bin(2k-1,p)<k]$ when $p<1/2$.  Note also that
$\pr[\bin(2k-2,p)<k]\ge \pr[\bin(2k-1,p)<k]$.

Now let $k = \ceil{\frac{\sizeof{v} + 1}{2}}$, so $k \geq 2$ and
$|v| \in\{2k-2,2k-1\}$.  For the event that $I$ has fewer than $k$ elements of
$v$, our observations about the binomial distribution yield
\[\pr[A_v]\ge\pr[\bin(2k-1,p)<k]\ge\pr[\bin(3,p)<2]=(1-p)^2(1+2p).\]

Let $q = \min_v \pr[A_v]$.  Events of the form $A_v$ are down-sets in the
subset lattice.  By the FKG inequality (see Theorem 6.2.1 of Alon and
Spencer~\cite{AS}), such events are positively correlated when $p<1/2$, so
  \[ \pr\left[\bigcap_{v \in S}A_v\right] \geq q^t = e^{t\ln q}. \]
Now let $X = \sizeof{I}$.  For $m\le \alpha tp$ with $\alpha<1$, Chernoff's
Inequality yields
  \begin{align*}
    \pr[X < m] &= \pr[X - tp < m - tp] \le
     e^{-(m-tp)^2/(2tp)} = e^{-(1-\alpha)^2tp/2}.
  \end{align*}
Our goal is to show $\pr\left[\bigcap_{v \in S}A_v\right]>\pr[X < m]$,
which follows from
$$\ln[(1-p)^2(1+2p)] > -(1-\alpha)^2p/2.$$
With $\alpha=.324722$ and $p=.079532$, the strict inequality holds, and we
obtain $\alpha p\approx .0258259$.  Hence when $d\ge m/(\alpha p)\ge 38.73m$, 
some $m$-set avoids all vertices in $S$.
\end{proof}

Before we apply this lemma to the game on the hypercube, we prove a general
result that relates the game on a graph and its retracts.  The notion of 
retract appeared as early as Hell~\cite{H}, as a homomorphism fixing a
subgraph.  The variation from~\cite{NW} that we use becomes the homomorphism
version when loops are available at all vertices.

\begin{defin}\label{retrdef}
An induced subgraph $H$ of a graph $G$ is a \emph{retract} of $G$ if there is a
map $f\st V(G) \to V(H)$ such that (1) $f(v)=v$ for $v\in V(H)$, and
(2) $uv\in E(G)$ implies that $f(u)$ and $f(v)$ are equal or adjacent.
\end{defin}

Nowakowski and Winkler~\cite{NW} proved a theorem for the classical
cop-and-robber pursuit game that is analogous to our next result.
 
\begin{thm}\label{retract}
Let $H$ be a retract of a graph $G$. If the \revs\ win $\RS(H,m,r,s)$, then the
\revs\ win $\RS(G,m,r,s)$.  Equivalently, $\sigma(G,m,r) \geq \sigma(H,m,r)$.
\end{thm}
\begin{proof}
Let $f\st G \to H$ be as guaranteed in Definition~\ref{retrdef}.  The \revs\
play in $G$ by playing exclusively on $H$, using the map $f$ to play as if the
spies in $V(G)-V(H)$ were actually in $V(H)$. 

The \revs\ take initial positions as specified by their winning strategy on $H$.
They simulate a spy on $v \in V(G)$ by a spy on $f(v) \in V(H)$. Whenever a
spy can legally move from $u$ to $v$ in $G$, the definition of retract
guarantees that the simulated spy can move from $f(u)$ to $f(v)$ in $H$.
Therefore, the simulated spies always play legal moves in the imagined game.
The \revs\ play their winning strategy against the simulated spies in $H$ and
eventually form an uncovered meeting at some vertex $w$.  Since $f(w)=w$, the
absence of a simulated spy on $w$ means that there is no real spy on $w$, and
the revolutionaries have won the ``real game'' in $G$.
\end{proof}

\begin{thm}\label{hyplinear}
If $s\le r-38.73m$ and $d\ge r$, then the revolutionaries win $\RS(Q_d,m,r,s)$.
\end{thm}
\begin{proof}
The \revs\ initially occupy $v_{1}, \ldots, v_{r}$.
The \revs\ threaten meetings after $m-1$ steps at $\CH rm$ vertices of weight
$m$.  The vertices of weight $m$ protected by a spy at $v_{i}$ are
precisely those whose corresponding sets contain $i$.  Let $t$ be the number of
\revs\ uncovered after the initial placement of spies.  By symmetry, we may
assume that the uncovered \revs\ are at $v_{1}, \ldots, v_{t}$.  Let
$S$ be the set of spies initially on vertices having weight at least $2$; only
such spies can protect vertices in the set of $\CH tm$ vertices of weight $m$
above uncovered \revs.  Note that $0\le |S|\le s-(r-t)\le t-38.73m$, and hence
$t\ge 38.73m$.


Every subcube of $Q_d$ is a retract of $Q_d$, by projection.  Hence by
Theorem~\ref{retract}, we may assume that the spies in $S$ are all in $Q_t$.
We can therefore apply Lemma~\ref{lem:prob}.  With $t\ge 38.73m$ and
$|S|\le t-38.73m< t$, some vertex of weight $m$ in $Q_t$ is too far from $S$ to
be reached by any spy within $m-1$ rounds, and the \revs\ win.
\end{proof}

Although $|S|\le t-38.73m$ in Theorem~\ref{hyplinear} while
Lemma~\ref{lem:prob} allows $|S|\le t$, generalizing the lemma to vary
$|S|$ in terms of $t$ does not noticeably strengthen the application.

When $t\ge 2m$, an explicit counting bound on the number of vertices of weight
$m$ in $Q_t$ that are within distance $m-1$ of a given vertex of $S$ leads
to the following theorem.

\begin{thm}\label{hypercubem}
If $d\ge r \geq m\ge 3$ and $s\le r-\FR34m^2$, then the revolutionaries win
$\RS(Q_d,m,r,s)$, so $\sigma(Q_d,m,r)> r-\FR34m^2$.
\end{thm}

Theorem~\ref{hypercubem} is stronger than Theorem~\ref{hyplinear} when
$m\le 52$.  We omit the proof, because the proofs of this counting lemma and
theorem are longer and more technical than those of Lemma~\ref{lem:prob} and
Theorem~\ref{hyplinear}, and because we believe that the \revs\ may win against
as many as $r-2m$ spies.

As in Theorem~\ref{hypercube}, the \revs\ in Theorem~\ref{hyplinear} play
locally, winning by staying within distance $m$ of a fixed vertex.  Hence with
general meeting size $m$ we can apply the same coding theory argument as in
Corollary~\ref{smallhyp}.  Given a code with distance $4m-1$, the balls of
radius $2m-1$ are disjoint.  Any vertex with distance more than $2m-1$ from the
central point has distance more than $m-1$ from the threatened meetings and
cannot reach them in $m-1$ turns, which is the number of rounds the \revs\
need to win in the strategy of Theorem~\ref{hyplinear}.  We thus have the
following.

\begin{cor}
If $d < r\le 2^d/d^{4m}$, then $\sigma(Q_d, m, r) > (d-38m)\FL{r/d}$.
\end{cor}

Finally, the hypercube result applies to more general cartesian products via
the notion of retract.  For $U\esub V(G)$, we use $G[U]$ to denote the subgraph
of $G$ induced by $U$.

\begin{cor}
Let $G = G_1\cart\cdots \cart G_d$, where $G_1, \ldots, G_d$ are graphs with at
least one edge.  If the revolutionaries win $\RS(Q_d, m, r, s)$, then the
revolutionaries win $\RS(G, m, r, s)$.
\end{cor}
\begin{proof}
By Theorem~\ref{retract}, it suffices to show that $G$ contains a retract
isomorphic to $Q_d$.  Select $v_iw_i \in E(G_i)$ for each $i$, and let
$U = \{v_1,w_1\}\times \cdots \times \{v_d,w_d\}$.  Note that $G[U] \cong Q_d$. 

To define $f\st V(G)\to U$, first define $g_i\st V(G_i) \to \{v_i, w_i\}$ by
setting  $g_i(x) =v_i$ if $x=v_i$ and $g_i(x)=w_i$ otherwise.
Now let $f(\VEC x1d) = (g_1(x_1), \ldots, g_d(x_d))$.  Clearly $f$ fixes $U$.
If $xy\in E(G)$, then there exists exactly one $i$ such that $x_i \neq y_i$;
without loss of generality, $x_i \neq v_i$. If also $y_i \neq v_i$, then
$g_i(x_i) = g_i(y_i) = w_i$, so $f(x) = f(y)$.

On the other hand, if $y_i = v_i$, then $g_i(x_i) = w_i$ and $g_i(y_i) = v_i$
while $g_j(x_j) = g_j(y_j)$ for all $j \neq i$, so $f(x) f(y)\in E(G[U])$ since
$w_i v_i\in E(G)$. Therefore $f$ satisfies the conditions in
Definition~\ref{retrdef}, and $G[U]$ is a retract of $G$ isomorphic to $Q_d$.
\end{proof}

\section{Random Graphs}\label{randsec}

In the Erd\H{o}s--Renyi binomial model $G(n,p)$, the vertex set is $[n]$, pairs
of vertices occur as edges independently with probability $p$, and we say that
an event occurs {\it almost surely} if its probability tends to $1$ as
$n\to\infty$.

When the graph is randomly generated and there are not too many \revs, the
\revs\ can play a strategy like that in Proposition~\ref{split} to defeat $r-m$
spies: the \revs\ occupy vertices so that no matter where the spies are placed,
any $m$ uncovered vertices can meet at some vertex adjacent to no spy.
When the number of \revs\ is larger, also the allowed number of spies is larger;
the \revs\ no longer can find such a placement, and the number of spies 
needed is only a fraction of $r$.

Our main task in this section is to show that for constant edge-probability $p$,
these two situations for the number of \revs\ are surprisingly close together,
differing only by a constant factor.  In particular, when $r<\ln2\ln n$ the
\revs\ almost always win agains $r-m$ spies, and when $r>cm\ln n$ almost
always $cr/m$ spies can win, where $c$ is any constant greater than $4$.  The
argument in the first setting also yields results when $p$ depends on $n$.

Independently, Mitsche and Pra\l at~\cite{MP} have proved that for $G$ in
$G(n,p)$, almost surely $\sgmr\le \FR rm+2(2+\sqrt 2+\epsilon)\log_{1/(1-p)}n$;
here $p$ can depend on $n$ (they also obtain conditions under which $r-m+1$
spies are needed).  Their upper bound is sharp within an additive constant, but
also they require $r$ to grow faster than $(\log n)/p$.  In comparison to our
method, they use more intricate structural characteristics of the random graph
and a more complex strategy for the spies.  Our strategy for the spies is like
that used elsewhere in this paper: introduce a notion of ``stable position''
that keeps the meetings covered, and show that the spies can maintain a stable
position.

First we consider the range where $r-m+1$ spies are needed.  Motivated by Alon
and Spencer~\cite{AS}, we say that $G$ has the \emph{$r$-extension property} if
for any disjoint $T,U \subset V(G)$ with $\C T + \C U \leq r$, there is a
vertex $x \in V(G)$ adjacent to all of $T$ and none of $U$.  We first show
why this property makes the game easy for the \revs.

\begin{prop}
If a graph $G$ satisfies the $r$-extension property, and $m\le r'\le r$,
then $G$ is spy-bad for $r'$ \revs\ and meeting size $m$.
\end{prop}
\begin{proof}
The $r'$ \revs\ initially occupy any set of $r'$ vertices in $G$.  To see that
$r'-m$ spies cannot prevent them from winning on the first round, let $U$ be the
set occupied by the spies, and let $T$ be the set occupied by uncovered \revs.
The \revs\ on $T$ win by moving to the vertex $x$ guaranteed by the 
$r$-extension property.
\end{proof}

Alon and Spencer~\cite[Theorem~10.4.5]{AS} present the result below for
constant $r$, but the proof holds more generally.

\begin{thm}
Let $\epsilon = \min\{p, 1-p\}$, where $p$ is a probability that depends on $n$.
If $r = o\left(\frac{n\epsilon^r}{\ln n}\right)$ and $n\epsilon^r \to \infty$,
then $G(n,p)$ almost surely has the $r$-extension property (and hence is
spy-bad for all $m$ and $r'$ with $m\le r'\le r$).
\end{thm}
\begin{proof}
Let $G$ be distributed as $G(n,p)$.  Given $T,U \subset V(G)$ with
$\C T+\C U\le r$, write $t=\C T$ and $u=\C U$.  For $x\in V(G)-(T\cup U)$, let
$A_{T,U,x}$ be the event that $x$ is adjacent to all of $T$ and none of $U$;
note that $\pr[A_{T,U,x}] = p^t(1-p)^u \geq \epsilon^r$.

Let $A_{T,U}$ be the event that $A_{T,U,x}$ fails for all $x\in V(G)-(T\cup U)$.
The events $A_{T,U,x}$ for different $x$ are determined by disjoint sets of 
vertex pairs, so
$\pr[A_{T,U}] \leq (1-\epsilon^r)^{n-r} \leq e^{-\epsilon^r(n-r)}$.

The $r$-extension property fails if and only if some event of the form
$A_{T,U}$ occurs.  Hence it suffices to show that the probability of their
union tends to $0$.  There are $3^r$ ways to form $T$ and $U$ within a fixed
$r$-set of vertices, since a vertex can be in either set or be omitted, and
there are $\CH nr$ sets of size $r$.  Hence the union consists of at most 
$(3n)^r$ events, each of whose probability is at most $e^{-\epsilon^r(n-r)}$.
We compute
$$
(3n)^re^{-\epsilon^r(n-r)} = e^{r\ln(3n) - \epsilon^r(n-r)}
= e^{r\ln 3 + r\ln n - \epsilon^r(n-r)}.
$$
Since $\epsilon\le 1/2$, the condition $r=o\left(\FR{n\epsilon^r}{\ln n}\right)$
implies $r = o(n)$, so the exponent is dominated by $-n\epsilon^r$ and tends to
$-\infty$.  Thus the bound on the probability of lacking the $r$-extension
property tends to $0$, and $G(n,p)$ almost surely satisfies this property.
\end{proof}

In particular, when $p$ is constant, $G(n,p)$ is almost surely spy-bad for
$r\ge m$ when $r\le c\ln n$, where $c<\ln(1/\epsilon)$.  Similarly, when $r$ is
constant, $G(n,p)$ is almost surely spy-bad when $p$ tends to $0$ more slowly
than $1/n^{1/r}$.  With $p\le 1/2$, the key condition is $np^r\to\infty$.

Now we confine our attention to the realm of constant edge-probability $p$
and consider well-known properties of the random graph that enable the spies to
do well.  For every vertex, the expected degree is $p(n-1)$, and for any two
vertices the expected size of their common neighborhood is $p^2(n-2)$.
Moreover, these random variables are so highly concentrated at their
expectations that almost always the degrees of \emph{all} vertices and the
sizes of common neighborhoods of \emph{all} pairs are within constant factors
of their expected values.  We begin by stating this formally; the proofs are
standard and straightforward using the Chernoff Bound.  We treat $G$ as a
sample from the model $G(n,p)$.

\begin{lem}\label{randeg}
Fix $p$ and $\gamma$ with $0<\gamma<p<1$.  In the random graph model $G(n,p)$,
almost surely $(p-\gamma)n<d(v)<(p+\gamma)n$ and
$(p^2 - \gamma^2)n<\sizeof{N(v)\cap N(w)}<(p^2+\gamma^2)n$
for all $v,w\in V(G)$.
\end{lem}

\begin{lem}\label{lem:alwayscommon}
Fix $p$ and $\gamma$ with $0<\gamma<p<1$.  In the random graph model $G(n,p)$,
almost surely $\frac{\sizeof{N(v) \cap N(w)}}{\sizeof{N(v)}} \geq p-\gamma$
for all $v,w\in V(G)$.
\end{lem}
\begin{proof}
Using the lower bound on common neighborhood size and the upper bound on
degree from Lemma~\ref{randeg}, almost surely 
$\FR{|N(v)\cap N(w)|}{|N(v)|}\ge\FR{(p^2-\gamma^2)n}{(p+\gamma)n}=p-\gamma$
for all $v,w\in V(G)$.
\end{proof}

\begin{defin}
For $q \in (0,1)$, a graph $G$ is \emph{$q$-common} if
$\frac{\sizeof{N(v) \cap N(w)}}{\sizeof{N(v)}} \geq q$ for all $v,w \in G$.
\end{defin}

We develop a strategy for spies that will be successful on $q$-common graphs
under certain conditions.  In a game position, we need to distinguish players
occupied in forming or covering meetings from those who are not.  These notions
will also be important for spy strategies on complete multipartite or bipartite
graphs.
 
\begin{defin}\label{stablenot}
Given a game position, say that $m$ specified \revs\ in a meeting and one spy
covering them are \emph{bound}.  After designating the bound players for all
vertices hosting meetings, the remaining spies and \revs\ are \emph{free}.
A vertex having at least $m$ \revs\ has exactly $m$ bound \revs.

For a vertex subset $U$, let $r_U$ and $\hat r_U$ denote the total number of
\revs\ and number of free \revs\ on $U$.  Similarly, let $s_U$ and $\hat s_U$
denote the total number of spies and number of free spies on $U$.
Write $\hat r$ and $\hat s$ for $\hat r_{V(G)}$ and $\hat s_{V(G)}$.
A game position is \emph{stable} if (1) all meetings are covered, and
(2) $\hat s_{N[v]}\ge \hat r/m$ for all $v\in V(G)$.
\end{defin}

As in Section~\ref{spygoodsec}, the name \emph{stable} is motivated by 
permitting the game to continue.

\begin{lem}\label{lem:cover}
On any graph $G$, if the position at the beginning of a round is stable, then
the spies can respond to cover all meetings at the end of the round.
\end{lem}
\begin{proof}
Let the notation in Definition~\ref{stablenot} refer to the counts at the
beginning of round $t$, in a stable position.  Let $X$ be the set of distinct
vertices hosting meetings after the \revs\ move in round $t$.  Let $Y$ be the
set of spies.  Define an auxiliary bipartite graph $H$ with partite sets $X$
and $Y$.  For $x \in X$ and $y \in Y$, put $xy\in E(H)$ if spy $y$ can reach
$x$ from its position at the start of round $t$, being adjacent to $x$ or
already there.  If some matching in $H$ covers $X$, then the spies can move in
round $t$ to cover all the meetings.

It suffices to show that $H$ satisfies Hall's~Condition for a matching that
covers $X$.  Consider $S\esub X$.  If $N_G[S]$ contains $b$ vertices that
hosted meetings at the start of round $t$, then $|S|\le \frac{\free{r}+mb}{m}$
meetings, because \revs\ who were in meetings not in $N_G[S]$ cannot reach
$S$ in one move.  On the other hand, every free spy at a vertex of $N_G[S]$
can reach $S$ in one move, as can every spy bound to a meeting in $S$.
Choosing $x\in S$, we have
$$|N_H(S)|\ge \hat s_{N[v]}+b\ge \hat r/m+b\ge |S|.$$
Hence Hall's Condition is satisfied and the matching exists.
\end{proof}

The next lemma provides the second half of what the spies need to do.

\begin{lem}\label{lem:makegood}
Let $G$ be a $q$-common graph with $n$ vertices, and fix $\epsilon > 0$.
Given a position in $\RSG$ such that (1) all meetings are covered,
(2) $\free{s} \geq \frac{1+\epsilon}{q}\frac{\free{r}}{m}$, and
(3) $\free{s} \geq \frac{\ln n}{2(1 - 1/(1 + \epsilon))^2q^2}$,
the free spies can move to produce a stable postion.
\end{lem}
\begin{proof}
We prove that if each free spy moves to a uniformly random vertex in the 
neighborhood of its current position, then with positive probability a stable
position is produced.

For $v\in V(G)$, let $X_v$ be the number of spies in $N[v]$ after the frees
spies move.  Since $G$ is $q$-common, each free spy lands in $N[v]$ with
probability at least $q$.  Also, these events for individual spies are
independent, so $X_v$ is a sum of $\free{s}$ independent indicator variables,
each with success probability at least $q$.  By the Chernoff Bound,
$\pr[X_v - \ex[X_v] < -a] < e^{-2a^2/\free{s}}$ for any positive $a$.  Since
$\ex[X_v]\ge q\free{s}$, taking $a = \left(1-\FR 1{1+\epsilon}\right)q\free{s}$
yields
\[
\pr\left[X_v < \frac{1}{1+\epsilon}q\free{s} \right]
< e^{-2\left(1 - \frac{1}{1+\epsilon}\right)^2q^2\free{s}}
= e^{-\ln n} =1/n,
\]
where the simplification of the exponent uses hypothesis~(3).

Since $G$ has $n$ vertices, with positive probability each vertex receives
at least $\frac{1}{1+\epsilon}q\free{s}$ free spies in its neighborhood.
By condition (2), this quantity is at least $\hat r/m$.  Hence there is
some move by the free spies after which each closed neighborhood has at least 
$\hat r/m$ free spies, making the position stable.
\end{proof}

\begin{thm}\label{commonwin}
Let $G$ be a $q$-common graph with $n$ vertices, and fix $\epsilon > 0$.
If $s \geq \frac{1+\epsilon}{q}\frac{r}{m}$ and 
$s \geq \frac{r}{m} + \frac{\ln n}{2(1 - 1/(1 + \epsilon))^2q^2}$, then the
spies win $\RS(G,m,r,s)$.
\end{thm}
\begin{proof}
If they can producing a stable position via the initial placements, the spies
use the following strategy in each subsequent round to produce a stable
position.  In Phase~1, they cover all meetings by moving the fewest possible
spies.  In Phase~2, they move the spies who are then free to produce a stable
position.

Since every spy moved in Phase~1 covers a meeting (by the condition of moving
the fewest spies), this strategy never moves a spy twice in one round.  Since
the position at the beginning of the round is stable, Lemma~\ref{lem:cover}
implies that spies can move to cover all meetings.  Hence Phase~1 can be
performed.  (Also, in the initial placement the spies can start by covering all
meetings, since $s\ge r/m$.)

If $\hat s$ is now large enough to satisfy the hypotheses of 
Lemma~\ref{lem:makegood}, then the free spies can complete Phase~2.
This argument is also used to complete the initial placement: after covering
the initial meetings, the free spies imagine being at an arbitrary vertex, and
then Lemma~\ref{lem:makegood} guarantees that they can ``move'' (that is, be
placed) to satisfy the neighborhood requirement for stability.

Consider the position after Phase~1; all meetings are covered.  Since at most
$r/m$ spies can be bound, the second assumed lower bound on $s$ yields
$\free{s}~\ge~s - \frac{r}{m} ~\ge~ \FR{\ln n}{2(1-1/(1+\epsilon))^2q^2}$.

Finally, we use the given lower bound $s\le \FR{1+\epsilon}q\FR rm$ to obtain
the needed lower bound $\hat s\le \FR{1+\epsilon}q\FR{\hat r}m$ that completes
the hypotheses of Lemma~\ref{lem:makegood}.  Let $\mtg r$ denote the number
of bound \revs\ at the start of the round.  Since $q<1<1+\epsilon$, we have
$$
\free{s} ~=~ s-\frac{\mtg{r}}{m} 
~\ge~ \frac{1+\epsilon}{q}\frac{r}{m} - \frac{1+\epsilon}{q}\frac{\mtg{r}}{m}
~=~ \frac{1+\epsilon}{q}\frac{\free{r}}{m}.
$$
We have shown that Phase~1 and Phase~2 can be completed to maintain a stable
position after each round.
\end{proof}

\begin{thm}\label{thm:random}
Fix $p$ and $q$ with $0<q<p<1$.  In the random graph model $G(n,p)$, almost
always $G$ has the following property for all $m \in \NN$: if
$s\ge\FR{1+\epsilon}q\FR rm$ and
$s\ge\FR rm + \FR{\ln n}{2(1-1/(1+\epsilon))^2q^2}$,
then the spies win $\RS(G,m,r,s)$.
\end{thm}
\begin{proof}
By Lemma~\ref{lem:alwayscommon}, almost always $G$ is $q$-common.
By Theorem~\ref{commonwin}, the spies win in the given parameter range on every
$q$-common graph.
\end{proof}

Since $1/q>1$, the next hypotheses imply the hypotheses of
Theorem~\ref{thm:random}.

\begin{cor}
For $p,q,G$ as above, almost surely $G$ has the following property for all
$m\in\NN$: if $s\ge\frac{1+\epsilon}{q}\FR rm$ and
$r\ge \FR{(1+\epsilon)^2 m\ln n}{2\epsilon^3q}$,
then the spies win $\RS(G,m,r,s)$.
\end{cor}

In particular, for the random graph with $p=1/2$, setting $\epsilon=1$ and
letting $q$ approach $1/2$ from below yields the following simply-stated
corollary.
\begin{cor}
Almost every graph $G$ has the following property for all $m\in \NN$
and $c>4$: if $s \geq c\frac{r}{m}$ and $r \geq cm\ln n$, then the spies win
$\RS(G,m,r,s)$.
\end{cor}


For sparse graphs, as $p\to 0$, we also need $q\to 0$, and the needed number
of \revs\ to apply our method grows at a faster rate than $m\ln n$.
Hence for sparse graphs we do not obtain the conclusion that the ranges
for $r$ where the needed number of spies behaves like $cr/m$ or like $r-m$
are close together.

\section{Complete $k$-partite Graphs}\label{multisec}

In this section we obtain lower and upper bounds on $\sgmr$ when $G$ is a
complete $k$-partite graph.  The lower bound requires partite sets large enough
so that the \revs\ can always access as many vertices in each part as they
might want (enough to ``swarm'' to distinct vertices there that avoid all the
spies).  The upper bounds apply more generally; they do not require large
partite sets, and they require only a spanning $k$-partite subgraph (if there
are additional edges within parts, then spies will be able to follow \revs\
along them when needed).

\begin{defin}
A complete $k$-partite graph $G$ is \emph{$r$-large} if every part has at least
$2r$ vertices.  At the \revs' turn on such a graph, an \emph{$i$-swarm} is a
move in which the \revs\ make as many new meetings of size $m$ as possible in
part $i$.  All \revs\ outside part $i$ move to part $i$, greedily filling
uncovered partial meetings to size $m$ and then making additional meetings of
size $m$ from the remaining incoming \revs.  When $G$ is $r$-large, sufficient
vertices are available in part $i$ to permit this.
\end{defin}

\begin{thm}\label{kpartlower}
Let $G$ be an $r$-large complete $k$-partite graph.  If $k\ge m$, then 
$\sgmr\ge \FR k{k-1}\FR{k\FL{r/k}}{m+c}-k$, where $c=1/(k-1)$.  When
$k\mid r$ the bound simplifies to $\FR k{k-1}\FR r{m+c}-k$.
\end{thm}
\begin{proof}
We may assume that $k\mid r$, since otherwise the \revs\ can play the strategy
for the next lower multiple of $k$, ignoring the extra \revs.

Let $t=r/k$.  The \revs\ initially occupy $t$ distinct vertices in each part.
Let $s_i$ be the initial number of spies in part $i$.  We may assume that they
cover $\min\{s_i,t\}$ distinct \revs, since each vertex of part $i$ has the
same neighborhood, and within part $i$ these are the best locations.
We compute the number of spies needed to avoid losing by a swarm on round $1$.

{\bf Case 1:} {\it $s_i>t$ for some $i$.}
If the \revs\ swarm to part $i$, then all \revs\ previously in part $i$ are
covered, so new meetings consist entirely of incoming \revs\ and are not
coverable by spies from part $i$.  Since $(k-1)t$ \revs\ arrive, at least
$\FL{(k-1)t/m}$ spies must arrive from other parts to cover the new meetings.
Thus 
$$
s\ge s_i+\FL{\FR{(k-1)t}m}\ge t\left(1+\FR{k-1}m\right)=\FR{k-1+m}k\,\FR rm.
$$

{\bf Case 2:} {\it $s_i\le t$ for all $i$.}
For each $i$, part $i$ has $t-s_i$ partial meetings.  Since $s_i\ge 0$, an
$i$-swarm is guaranteed to fill them if $(k-1)t\ge t(m-1)$, which holds when
$k\ge m$.  Hence the new meetings include all \revs\ except the $s_i$ covered
by spies in part $i$ before the swarm.  Spies from other parts must cover
$\FL{(r-s_i)/m}$ new meetings in part $i$.  Summing $s-s_i\ge(r-s_i-m+1)/m$
over all parts yields $(k-1+1/m)s\ge k(r-m+1)/m$, so
$$
s\ge \FR{k(r-m+1)}{m(k-1)+1} > \FR k{k-1}\FR r{m+c}-k.
$$

The lower bound in Case 2 is smaller (better for spies) than the lower bound
in Case 1, so the spies will prefer to play that way.  The lower bound in 
Case 2 is thus a lower bound on $\sgmr$.
\end{proof}

%

As in Section~\ref{randsec}, our strategy for spies maintains a ``stable
position'', defined by invariants ensuring that the spies can cover all
meetings and reestablish a stable position.  Indeed, for complete multipartite
graphs the notion of stable position is very similar to what it was in the
random graph.

\begin{defin}\label{stablemult}
Define bound and free \revs\ and spies as in Definition~\ref{stablenot}.
Let $\hat r_i$ and $\hat s_i$ denote the numbers of free \revs\ and free spies
in part $i$ in the current position of a game on a complete $k$-partite graph.
Let $\hat r$ and $\hat s$ denote the total numbers of free \revs\ and free
spies.  A game position is \emph{stable} if (1) all meetings are covered, and
(2) $\hat s-\hat s_i\ge \hat r/m$ for each part $i$.
\end{defin}

Since the neighborhood of a vertex in a complete multipartite graph consists
of all the partite sets not containing it, for such a graph $G$ the condition
for a stable position is the same as it was in Section~\ref{randsec}.

\begin{lem}\label{stable}
Let $G$ be a graph having a spanning complete $k$-partite subgraph $G'$.  If
the position at the start of round $t$ is stable for $G'$, then the \revs\
cannot win in the current round on $G$.  (As always, assume $s\ge \FL{r/m}$.)
\end{lem}
\begin{proof}
We follow the argument of Lemma~\ref{lem:cover} and hence summarize the
steps.  The designation of and notation for free and bound players is as of the
start of round $t$.  Let $X$ be the set of distinct vertices hosting meetings
after the \revs\ move in round $t$, let $Y$ be the set of spies, and let $H$ be
the bipartite graph $H$ with partite sets $X$ and $Y$ that encodes which spies
can move to cover which meetings.  

We show that $H$ satisfies Hall's~Condition for a matching that covers $X$.
Consider $S\esub X$.  Note that $|X|\le\FL{r/m}\le s$.  If $S$ has vertices
from more than one partite set in $G'$, then $|N_H(S)| = s \ge |X|$.

If $S$ has vertices only from part $i$ in $G'$, then we may assume that no
vertices of $S$ correspond to old meetings, since they would remain covered by
their bound spies.  Let $p$ be the number of vertices in $N_{G'}[S]$ hosting
meetings at the start of round $t$.  By stability, these vertices have bound
spies, which lie in $N_H(S)$.  Stability also guarantees
$\hat{s} - \hat{s_i} \geq \hat{r}/m$, and all of the free spies counted by
$\hat{s} - \hat{s_i}$ are also in $N_H(S)$.  No spy is both free and bound, so
$|N_H(S)| \geq p + \hat{r}/m$.  On the other hand, the number of \revs\ that
can be used to make meetings in $S$ is at most $\hat r+pm$, since only $m$
\revs\ at a vertex having a meeting are bound; the rest are free.  Hence
$|S|\le\FR{\hat{r}+pm}{m}\le |N_H(S)|$, as desired.
\end{proof}

\begin{thm}\label{kpartupper}
If a graph $G$ has a spanning complete $k$-partite subgraph, then
$\sgmr\le \CL{\FR k{k-1}\FR{r}{m}}+k$.
\end{thm}
\begin{proof}
Let $G'$ be the specified subgraph, and let $s=\CL{\FR k{k-1}\FR{r}{m}}+k$.
It suffices to show that $s$ spies can produce a stable position at the end of
each round.  First, after the \revs\ have moved, the spies cover all newly
created meetings, moving the fewest possible spies to do so.
By Lemma~\ref{stable}, the spies can do this since the previous round ended in
a stable position (also, $s\ge\FL{r/m}$ guarantees that the spies can do this
in the initial position).

Next, the spies that are now free distribute themselves equally among the $k$
parts of $G'$.  More precisely, with $\hat s$ being the total number of free
spies after the new meetings are covered and $\hat s_i$ being the number of
them in part $i$, we have $|\hat s_i-\hat s/k|<1$ for all $i$.

It suffices to show that this second step produces a stable position.
In order to have $\hat s-\hat s_i\ge \hat r/m$ for all $i$, it suffices to
have $\hat s_j\ge \hat r/[m(k-1)]$ for each $j$.  Since the free spies are
distributed equally, it suffices for the average to be big enough:
$\hat s/k\ge \hat r/[m(k-1)]+1$.  Multiplying by $k$, we require
$\hat s\ge \FR k{k-1}\FR{\hat r}m+k$.

We are given $s\ge \FR k{k-1}\FR{r}m+k$.  The number of bound \revs\ is exactly
$m$ times the number of bound spies; hence $s-\hat s=(r-\hat r)/m$.  Subtracting
this equality from the given inequality yields
$$
\hat s\ge \FR 1{k-1}\FR rm + \FR{\hat r}m+k \ge \FR k{k-1}\FR{\hat r}m+k,
$$
where the last inequality uses $r\ge \hat r$.  We now have the inequality that
we showed suffices for a stable position.
\end{proof}

\section{Complete Bipartite Graphs}\label{bicliqsec}
Finally, let $G$ be an $r$-large bipartite graph.  We give lower and upper
bounds on $\sgmr$ for fixed $m$.  The lower bounds use strategies for the
\revs\ that win after one or two rounds, while the upper bounds use more
delicate strategies for the spies (maintaining invariants that prevent the
\revs\ from winning on the next round).

Since the lower bounds are much easier, we start with them, but first we
compare all the bounds in Table 1.  When $3\mid m$, the lower bound is roughly
$\FR32 r/m$.  We believe that this is the asymptotic answer when $3\mid m$.
When $3\nmid m$, the \revs\ cannot employ this strategy quite so efficiently,
which leaves an opening for the spies to do better.  Indeed, for $m=2$, the
answer is roughly $\FR 75 r/m$, a bit smaller.  For larger $m$, the relative
value of this advantage diminishes, and we expect the leading coefficient to
tend to $3/2$ as $m\to\infty$.


\begin{table}[hbt]
\begin{center}
\caption{Bounds on $\sgmr$}\label{table}
\smallskip
\begin{tabular}{c|c|c|l}
Meeting size&Lower bound&Upper bound&References\\
\hline
$2$&$\bipmtwo$&$\bipmtwo$ &Theorems \ref{revm2} and \ref{spym2}$\hispace$\\
$3$&$\FL{r/2}$&$\FL{r/2}$&Theorems \ref{bipartm3} and \ref{spym3}\\
$m\in\{4,8,10\}$&$\FR15\FL{\FR{7r}m-\FR{13}2}$&$
$&Corollary~\ref{evenm}$\hispace$\\
$m$&$\bipmgen$&$\left(1+\FR1{\sqrt3}\right)\FR rm+1$
   &Corollary \ref{revmgen}; Theorem \ref{spymgen}
\end{tabular}
\end{center}
\end{table}

We first motivate the lower bounds by giving simple strategies for the \revs\
when $m\in\{2,3\}$.  Henceforth call the partite sets $X_1$ and $X_2$.

\begin{examp}\label{simprev}
\rm
Initially place $\FL{r/2}$ \revs\ in $X_1$ and $\CL{r/2}$ \revs\ in $X_2$.
Regardless of where the spies sit, swarming \revs\ can form at least
$\FL{(r-1)/(2m)}$ new meetings on either side that can only be covered by
spies from the other side, so the initial placement must satisfy
$s_1\ge \FL{(r-1)/(2m)}$ and $s_2\ge \FL{r/(2m)}$, where $s_i$ is the number of
spies in $X_i$.

However, the uncovered \revs\ can also be used to form meetings.  If $m=2$,
then the \revs\ can form $\FL{(r-s_i)/2}$ meetings when swarming to $X_i$,
so the spies lose unless $s_{3-i}\ge\FL{(r-s_i)/2}$ for both $i$.  Summing
the inequalities yields $s_1+s_2\ge 2(r-1)/3$.

For $m=3$, considering only $r$ of the form $4k$, where $k\in\NN$, we show
that the \revs\ win against $2k-1$ spies.  Initially there are $2k$ \revs\ in
each part, on distinct vertices.  We may assume $s_1\le s_2$, so $s_1\le k-1$.
Since there are only $2k-1-s_1$ spies in $X_2$, there are at least $s_1+1$
uncovered \revs\ in $X_2$.  Since $s_1\le k-1$, we can use $2(s_1+1)$ \revs\
from $X_1$ to form meetings of size $3$ with the uncovered \revs\ in $X_2$.
Since only $s_1$ spies are available to cover these meetings, the spies lose.
\looseness-1

Thus $\sghr\ge r/2$ when $4\mid r$.  However, when $r=4k+2$, the \revs\ cannot
immediately win against $2k$ spies by this construction.  With $2k+1$ \revs\
in each part and $k$ spies sitting on \revs\ in each part, swarming \revs\ can
only make $k$ new meetings in either part, which can be covered by the spies.
\qed
\end{examp}

The symmetric strategy in Example~\ref{simprev} is optimal when $m=3$ and
$4\mid r$.  However, when $m=2$ and when $m=3$ with $r=4k+2$, the \revs\ can
do better using an asymmetric strategy that takes advantage of moving away
from spies.  When $m=3$ and $r=4k+2$, this other strategy just increases the
threshold by $1$, to the value $\FL{r/2}$ that we will show is optimal for all
$r$.  For $m=2$, however, the better strategy increases the leading term from
$2r/3$ to $7r/10$.

Recall that the partite sets are $X_1$ and $X_2$ and that a vertex (or meeting)
is \emph{covered} if there is a spy there.  Say that a spy is \emph{lonely}
when at a vertex with no \rev.

\begin{thm} \label{revm2}
If $G$ is an $r$-large complete bipartite graph, then
$\sgtr \ge \bipmtwo$.
\end{thm}
\begin{proof}
We present a strategy for the \revs\ and compute the number of spies needed to
resist it.  The \revs\ start at $r$ distinct vertices in $X_1$.  In response,
at least $\FL{r/2}$ spies must start in $X_1$, since otherwise the \revs\ can
next make $\FL{r/2}$ meetings at uncovered vertices in $X_2$ and win.

In the first round, $\FL{r/2}$ \revs\ move from $X_1$ to $X_2$, occupying
distinct vertices.  They leave from vertices of $X_1$ that are covered by spies
(as much as possible), so after they move at least $\FL{r/2}$ spies in $X_1$
are lonely.  Now the spies move; let $s_i$ be the number of spies in $X_i$
after they move (for $i\in\{1,2\}$).  Let $c$ be the number of \revs\ in $X_1$
that are now covered by spies.  Since at most $s_2$ spies leave $X_1$, there
remain at least $\FL{r/2}-s_2$ lonely spies in $X_1$.  We conclude that
$c\le s_1-\FL{r/2}+s_2$.

In round 2, the \revs\ have the opportunity to swarm to $X_1$ or $X_2$.  Since
there are $\FL{r/2}$ \revs\ in $X_2$, there are at most $\FL{r/2}+1$ uncovered
\revs\ in $X_1$ (on distinct vertices), so swarming \revs\ can make meetings
with all but at most 1 uncovered \rev\ in $X_1$.  The \revs\ can therefore make
$\FL{(r-c)/2}$ new meetings in $X_1$.  These meetings can only be covered by
spies moving from $X_2$, so the spies lose unless $s_2\ge\FL{(r-c)/2}$.

If the \revs\ swarm to $X_2$, then the new meetings there can only be covered
by spies coming from $X_1$.  At most $s_2$ \revs\ in $X_2$ are covered by
spies.  Since $\CL{r/2}$ \revs\ come from $X_1$, they can make meetings with all
uncovered \revs\ in $X_2$, so the spies lose unless $s_1\ge \FL{(r-s_2)/2}$.

Adding twice the lower bound on $s_1$ to the lower bound on $s_2$ (with
$c\le s_1-\FL{r/2}+s_2$),
$$
s_2+2s_1\ge \FR{\FL{3r/2}-s_1-s_2-1}2+r-s_2-1.
$$
The inequality simplifies to $5(s_1+s_2)\ge \FL{7r/2}-3$, as desired.
\end{proof}

The general lower bound in Corollary~\ref{revmgen} uses the formula for $m=3$,
which we study first.  The key is that $r/2-1$ spies are not enough when
$r\equiv 2\mod 4$; we first sketch the idea in an easy case.  Suppose that
$r=4k+2\equiv 6\mod 12$.  The \revs\ start at distinct vertices in $X_1$.
Suppose that all $s$ spies start in $X_1$ and that there are enough of them
to win.  In round 1, $2r/3$ \revs\ move to $X_2$, leaving the spies in $X_1$
lonely.  Let $s_2$ be the number of spies that move to $X_2$ after round 1,
leaving $s_1$ spies in $X_1$.  The \revs\ in $X_2$ now can make $r/3$ meetings
with the remaining $r/3$ \revs\ in $X_1$, so $s_2\ge r/3$.  Since
$s_2\le 2k=r/2-1$, at least $r/6+1$ \revs\ remain uncovered in $X_2$.  The
remaining $r/3$ \revs\ in $X_1$ can make meetings with $r/6$ of them in round 2.
Hence $s_1\ge r/6$, and $s=s_1+s_2\ge r/2$.

The initial placement only requires $r/3$ spies in $X_1$, not $r/2$.  We must
allow for initial placement of $x$ spies in $X_2$, where $0\le x\le r/6$.
The $x$ spies originally in $X_2$ can move to $X_1$ in round 1 and cover \revs\
there; this prevents the \revs\ from threatening as many meetings by a swarm to
$X_1$.  In response, fewer than $2r/3$ \revs\ move to $X_2$ in round 1, and
yet we can guarantee more threatened meetings in the swarm to $X_2$.

\begin{thm}\label{bipartm3}
If $G$ is an $r$-large complete bipartite graph, then $\sghr\ge\FL{r/2}$.
\end{thm}
\begin{proof}
Since $\FL{r/2}=\FL{(r+1)/2}$ when $r$ is even, and having an extra \rev\ 
cannot reduce $\sigma$, it suffices to prove the lower bound when $r$ is even.
Example~\ref{simprev} proves it when $4\mid r$, so only the case $r=4k+2$
remains.  We show that $4k+2$ \revs\ can win against $2k$ spies.  Suppose that
the spies can survive for two full rounds after the initial placement.

The \revs\ start at $r$ distinct vertices of $X_1$, so at least $\FL{r/3}$
spies must start in $X_1$.  Let $x$ be the initial number of spies in $X_2$,
with $2k-x$ spies in $X_1$.  Since $X_1$ contains at least $\FL{r/3}$ spies,
$x\le \CL{(2k-2)/3}=\CL{r/6}-1$.  Define $j$ by $r-x\equiv j\mod 3$ with
$j\in\{0,1,2\}$.
In round 1, $p$ \revs\ move to $X_2$, where $p={2(r-x-j)/3}$.  Note that
$p\ge 2k-x$, so all spies in $X_1$ are now lonely.  The number of \revs\
remaining in $X_1$ is $r-p$, which equals ${(r+2x+2j)/3}$. 

Let $s_i$ be the number of spies in $X_i$ after the spies respond in round 1.
Since at most $x$ spies move from $X_2$ to $X_1$ in round 1, the number of
uncovered \revs\ in $X_1$ is now at least ${(r-x+2j)/3}$.  With
$p={2(r-x-j)/3}$, there are enough \revs\ in $X_2$ to threaten meetings at
${(r-x-j)/3}$ vertices in $X_1$ with \revs\ who remained there.  Hence
$s_2\ge {(r-x-j)/3}$.

Now consider a swarm to $X_2$ in round $2$.  Since there were $2k-x$ spies in
$X_1$ initially, the number who moved to $X_2$ and covered \revs\ after round
$1$ is at most $2k-x$.  Hence round $2$ starts with at least $p-2k+x$ uncovered
\revs\ in $X_2$.  The $r-p$ \revs\ remaining in $X_1$ move in pairs to generate
meetings with uncovered \revs\ in $X_2$.  Note that $r-p=(r+2x+2j)/3$ and
$p-2k+x=(r+2x+6-4j)/6$.  The number of meetings that can be made in $X_2$ (and
can only be covered by the $s_1$ spies in $X_1$) depends on $j$.

When $j=0$, the number of meetings made is $(r+2x)/6$, so $s_1\ge (r+2x)/6$,
and we obtain $s_2+s_1\ge \FR{r-x}3+\FR{r+2x}6=r/2$.  When $j=1$, the \revs\
can make $p-2k+x$ meetings in the swarm; hence $s_1\ge (r+2x+2)/6$, and we
obtain $s_2+s_1\ge \FR{r-x-1}3+\FR{r+2x+2}6=r/2$.  Finally, when $j=2$,
the same computation yields only $s\ge \FR{r-x-2}3+\FR{r+2x-2}6=r/2-1$.
However, equality holds only if all $2k-x$ spies initially in $X_1$ move to
$X_2$ in round $1$ to cover \revs.  Only $x$ spies remain in $X_1$ to guard
the swarm to $X_2$ that makes $(r+2x-2)/6$ meetings.  The inequality
$x\ge(r+2x-2)/6$ requires $x\ge (r-2)/4$, but guarding the initial position
required $x<r/6$.  
\end{proof}

\begin{cor}\label{revmgen}\label{evenm}
If $G$ is an $r$-large complete bipartite graph, then $\sgmr\ge\bipmgen$.
If $m$ is even, then $\sgmr\ge \FR15\FL{\FR{7r}m-\FR{13}2}$.
\end{cor}
\begin{proof}
Let $m'=\CL{m/3}$.  The \revs\ group into cells of size $m'$; each cell moves
together, modeling one player in a game with meeting size $3$.  When three of
these cells converge to make an unguarded meeting, the \revs\ win the original
game.  The $r$ \revs\ make $\FL{r/m'}$ such cells and ignore extra \revs.  By
Theorem~\ref{bipartm3}, the number of spies needed to keep the \revs\ from
winning is at least $\FL{\FL{r/m'}/2}$.

For even $m$, let $m'=m/2$.  The \revs\ can group into $\FL{r/m'}$ cells of
size $m'$ and play a game with meeting size $2$.  In the lower bound of
Theorem~\ref{revm2}, we replace $r$ by the number of cells in this imagined
game, which is $\FL{2r/m}$.  Dropping the outer ceiling function, the resulting
lower bound is $\FR15\FL{\FR72\FL{\FR{2r}m}-3}$.  We use
$\FL{\FR{2r}m}>\FR{2r}m-1$ to obtain the slightly simpler expression claimed.
It improves on the bound above when $m\in\{4,8,10\}$.
\end{proof}

Finally, we consider upper bounds for $\sgmr$ when $G$ is an $r$-large
bipartite graph, proved by giving strategies for the spies.

\begin{defin}\label{stuff}
Henceforth, always $G$ is an $r$-large bipartite graph with partite sets $X_1$
and $X_2$, and we consider the game $\RSG$.  Any statement that includes index
$j$ is considered for both $j=1$ and $j=2$.  The numbers of \revs\ and spies in
part $j$ at the beginning of the current round are denoted by $r_j$ and $s_j$,
respectively, and the number of \revs\ in part $j$ that are on vertices covered
by spies is denoted by $c_j$.  The corresponding counts at the end of the 
round are denoted by $r'_j$, $s'_j$ and $c'_j$.

A spy that moves to $X_j$ during the round is \emph{new}; spies that remained
in $X_j$ and did not move are \emph{old}.  A meeting formed at a vertex in
$X_j$ during the round is \emph{new} if at the end of the previous round there
was no meeting there; a meeting is \emph{old} if it is not new.  The \revs\
\emph{swarm} $X_j$ in a round if at the end of the round all \revs\ are in
$X_j$.
\end{defin}

\begin{defin}\label{greedymig}
\rm
A \emph{greedy migration strategy} is a strategy for the spies having the
following properties.  First, no vertex ever has more than one spy on it.
Next, after the \revs\ move during the current round and the spies compute
the new desired distribution $s_1',s_2'$ of spies on $X_1$ and $X_2$, they
move to reach that distribution as follows.  Since $s_1'+s_2'=s_1+s_2$, by
symmetry there is an index $i\in\{1,2\}$ such that $s_i'\le s_{3-i}$.  The
spies reach their locations for the end of the round via the following steps.

(1) $s_i'$ spies move away from $X_{3-i}$, iteratively leaving vertices that
now have the fewest \revs\ among those in $X_{3-i}$.

(2) {\it All} $s_i$ spies previously on $X_i$ leave $X_i$ and move to uncovered
vertices in $X_{3-i}$, iteratively covering vertices having the most \revs.

(3) The $s_i'$ spies that left $X_{3-i}$ now move to uncovered vertices in
$X_{i}$, iteratively covering vertices having the most \revs.
\end{defin}

\begin{rem}\label{allmove}
At the end of round $t$ under a greedy migration strategy, we designate each
meeting or spy as ``old'' or ``new''.  An {\it old meeting} is a meeting at a
vertex where there was also a meeting at the start of round $t$; all other
meetings at the end of round $t$ are {\it new}.  An {\it old spy} is a spy who
did not move during round $t$; all spies who moved are {\it new spies}.

For $j\in\{1,2\}$ either all spies that end round $t$ in $X_j$ are new
(started round $t$ in $X_{3-j}$), or all spies that started round $t$
in $X_{3-j}$ are new (end round $t$ in $X_j$).  In the specification of the
movements in Definition~\ref{greedymig}, the former occurs when $j=i$, and the
latter occurs when $j=3-i$.  In the first case, round $t$ ends with $s'_j$ new
spies in $X_j$; in the second case, it ends with $s_{3-j}$ new spies in $X_j$.
In particular, at least $\min\{s'_j,s_{3-j}\}$ spies in $X_j$ are new.
\end{rem}

\begin{lem} \label{fact:greedymigration}
A greedy migration strategy in $\RSG$ is a winning strategy for the spies
if it prevents the \revs\ from winning by swarming a part.
\end{lem}
\begin{proof}
As in Definition~\ref{stuff}, Let $r_j,s_j,r'_j,s'_j$ count the \revs\ and
spies at vertices of $X_j$ at the start and end of round $t$, respectively,
and define old and new meetings and spies as in Remark~\ref{allmove}.  We show
that if a given greedy migration strategy for the spies keeps the \revs\ from
winning by swarming on round $t$ or round $t+1$, then all meetings are covered
at the end of round $t$.  Hence the \revs\ never win.

By swarming to $X_{3-j}$ in round $t$, the \revs\ can produce at least
$\FL{r_j/m}$ new meetings there.  Since these meetings can be covered only by
spies in $X_j$ at the start of round $t$, and the strategy prevents the \revs\
from winning by this swarm, we obtain $s_j\ge\FL{r_j/m}$ (and similarly
$s_{3-j}\ge\FL{r_{3-j}/ m}$).  Applying the same argument in round $t+1$ yields
$s'_j\ge\FL{r'_j/m}$.

If all $s'_j$ spies in $X_j$ at the end of round $t$ are new, then they cover
all the meetings in $X_j$, since $s'_j\ge\FL{r'_j/m}$ and greedy migration
maximizes the coverage.  Hence we may assume that some of these $s'_j$ spies
are old.  Now Remark~\ref{allmove} implies that all $s_{3-j}$ spies in
$X_{3-j}$ at the start of round $t$ moved to $X_j$ during round $t$.  We
consider two cases:

{\bf Case 1:} \emph{In round $t$ every old meeting in $X_j$ is covered by some
old spy.}
In this case it remains to show that at the end of round $t$, the $s_{3-j}$
new spies in $X_j$ cover all the new meetings there.  We claim that otherwise
the \revs\ could have won in round $t$ by swarming to $X_j$.  A \rev\ who
stayed in $X_j$ or moved from $X_{3-j}$ to $X_j$ in the actual round $t$ also
would do so in a swarm to $X_j$.  A \rev\ who moved from $X_j$ to $X_{3-j}$
would instead remain in $X_j$ in the swarm, and a \rev\ who stayed in
$X_{3-j}$ in the actual round would move to a $X_j$ in in the swarm.
Thus the swarm produces at least as many new meetings, and the same number of
old meetings, as the \revs' actual moves in round $t$.  The spies therefore
cannot cover all of the new meetings formed by this swarm if their greeting
migration does not cover all of the new meetings actually formed in $X_j$ in
round $t$.

{\bf Case 2:} \emph{At the end of round $t$ some old meeting in $X_j$ is not
covered by an old spy}.  Since greedy migration picks departing spies to
minimize the number of revolutionaries uncovered, all old spies who remain in
$X_j$ are covering meetings.  The new spies who move to $X_j$ maximize
coverage, so if there is an uncovered meeting in $X_j$ at the end of
round $t$, then every spy in $X_j$ is covering a meeting.  Since
$s'_j\ge\FL{r'_j/m}$, all the meetings are covered.
\end{proof}

\begin{thm} \label{spym2}
If $G$ is an $r$-large complete bipartite graph, then $\sgtr\le\bipmtwo$.
\end{thm}
\begin{proof}
Let $s=\bipmtwo$; we give a winning strategy for the spies in $\RS(G,2,r,s)$.
Let $\alpha=s-\FL{r/2}$ and $\beta=\FL{(r-\alpha)/2}$.  Later we will use the
following inequalities: $\alpha\le\beta$, $\alpha+\beta\le s$, and
$\FL{(r+\beta)/2}\le s$.  These inequalities can be checked explicitly for
each congruence class modulo 10.  The first two are loose, since
$\alpha\approx 2r/10$, $\beta\approx 4r/10$, and $s\approx 7r/10$, but the
third is delicate, with equality holding except in two congruence classes
and the floor function needed for correctness in four congruence classes.

During the game, if the \revs\ swarm $X_{3-j}$ in the current round, then they
generate at most $\min\{r_j, \FL{r-c_{3-j}\over 2}\}$ new meetings.  The spy
strategy will ensure
$$
s_j\ge\min\left\{r_j,\FL{r-c_{3-j}\over2}\right\}\qquad
\textrm{for}~j\in\{1,2\},\eqno{(A)}
$$
and hence it will keep the \revs\ from winning by a swarm.  The spies move by
greedy migration after computing the new values $s_1'$ and $s_2'$ in response
to $r_1'$ and $r_2'$.  By Lemma~\ref{fact:greedymigration}, the spies win
by a greedy migration strategy that keeps the \revs\ from winning by swarm.

The spies determine $s'_1$ and $s'_2$ via three cases, using the first that
applies.  Always $s_1'+s_2'=s$.

\begin{description}
\item [Case 1:] If $r'_i\le\alpha$ for some $i\in\{1,2\}$, then $s_i'=\alpha$.

\item [Case 2:] If $s_i\ge \min\{r'_{3-i}, \beta\}$ for some $i\in\{1,2\}$,
then $s_{3-i}'=\min\{r'_{3-i}, \beta\}$.

\item [Case 3:] Otherwise, $s_i'=s_{3-i}$ and $s_{3-i}'=s_i$.
\end{description}

It remains to prove $(A)$.  In order to do so, we first prove
$$
s_j\ge\alpha\qquad\textrm{for}~j\in\{1,2\}.\eqno{(B)}.
$$
Trivially the spies can satisfy both $(A)$ and $(B)$ in round $0$.  Assuming
that these invariants hold before the current round begins, we will show that
they also hold when it ends.

\medskip
\emph{Invariant (B) is preserved.} 
In Case 1, $s'_i=\alpha$ and $s'_{3-i}=\floor{r/2}>\alpha$.
In Case 3, $s'_j=s_{3-j}\ge \alpha$.  In Case 2, $r'_{3-i}> \alpha$, so
$s'_{3-i}=\min\{r'_{3-i},\beta\}\ge\alpha$, and
$s'_i=s-s'_{3-i}=s-\min\{r'_{3-i},\beta\}\ge s-\beta\ge\alpha$.

\medskip
\emph{Invariant (A) is preserved.} 
In Case 1,
$s'_i=\alpha\ge r'_i\ge \min\{r'_i,\floor{r-c'_{3-i}\over2}\}$ and
$s'_{3-i}=\FL{r/2}\ge\FL{r-c'_{3-i}\over2}\ge\min\{r'_i,\FL{r-c'_{3-i}\over2}\}$.

In Case 2 with $s_i\ge\min\{r'_{3-i},\beta\}$, first consider $j=3-i$.  We have
$s'_{3-i}=\min\{r'_{3-i},\beta\}$.  If $s'_{3-i}=r'_{3-i}$, then $s'_{3-i}$ is
already big enough, so suppose $s'_{3-i}=\beta$.  By Remark~\ref{allmove},
at least $\min\{s'_i,s_{3-i}\}$ spies in $X_i$ are new.  By $(B)$, this
quantity is at least $\alpha$, and Case 2 requires $r'_i>\alpha$.  Hence the
new spies cover at least $\alpha$ \revs, and $c'_i\ge\alpha$ yields
$s'_{3-i}=\beta=\FL{\FR{r-\alpha}2}\ge\min\{r'_{3-i},\FL{\FR{r-c'_i}2}\}$.

Now consider $j=i$.  By Remark~\ref{allmove}, at least $\min\{s_i,s'_{3-i}\}$
spies in $X_i$ are new, and in Case 2 each of $s_i$ and $s'_{3-i}$ is at least
$\min\{r'_{3-i},\beta\}$.  Since spies cover greedily,
$c'_{3-i}\ge \min\{r'_{3-i},\beta\}=s'_{3-i}$.  Also $s'_{3-i}\le\beta$, so 
\begin{equation}\label{swarmguard}
s'_i=s-s'_{3-i}\ge \FL{\FR{r+\beta}2}-s'_{3-i}\ge \FL{\FR{r-s'_{3-i}}2}
\ge \FL{\FR{r-c'_{3-i}}2}\ge \min\left\{r'_i,\FL{\FR{r-c'_{3-i}}2}\right\}.
\end{equation}

Finally, $s'_j=s_{3-j}<\min\{r'_j,\beta\}$ in Case 3, since Case 2 does not
apply.  Since all spies move and $s'_j\le r'_j$, we have $c'_j\ge s'_j$.  Hence
for each $j$ the computation in (\ref{swarmguard}) is valid.
\end{proof}

The method for the upper bound when $m=3$ is essentially the same.

\begin{thm} \label{spym3}
If $G$ is an $r$-large complete bipartite graph, then $\sghr\le \FL{r/2}$.
\end{thm}
\begin{proof}
We present a greedy migration strategy for $\floor{r/2}$ spies that keeps the
\revs\ from winning by swarming; by \lemma{fact:greedymigration} it is a
winning strategy for the spies.

Define $r_j,s_j,c_j$ at the start of a round and $r'_j,s'_j,c'_j$ at the end of
the round in the same way as before.  Also, we need to know the maximum number
of \revs\ together on an uncovered vertex in $X_j$ at the beginning and end of
the round; let these values be $u_j$ and $u'_j$.  If the \revs\ have not
already won, then $u_j,u'_j\le 2$.  Let $s=\FL{r/2}$,
$\alpha=\FL{r/2}-\FL{r/3}$, and $\beta=s-\FL{(r-\alpha)/3}$.  We will want the
inequalities $\beta\ge\bFL{\FR{r-2\alpha}3}$ and
$\beta\le\bigceil{\FR{\FL{r/2}}2}$.  The latter is always satisfied (the left
side is about $2r/9$ and the right side is about $r/4$), but both sides of the
first inequality are about $2r/9$.  Checking each congruence class modulo $18$
shows that $\beta\ge\bFL{\FR{r-2\alpha}3}$ except when $r\equiv 3\mod 18$.

The values $s'_1$ and $s'_2$ that determine the movements of spies in this
round under the greedy migration strategy are computed as follows, with
$s'_{3-i}=s-s'_i$ always.  Note that since $r'_1+r'_2=r$, when one of the cases
below holds, it holds for exactly one index $i$ unless $r'_1=r'_2=r/2$.  In
this case of equality, it does not matter which index we call $i$.
\begin{description}
\item[Case 1:] If $r'_i\le\alpha$ for some $i\in\{1,2\}$, then
$s'_i=\alpha$.

\item[Case 2:] If $\alpha<r'_i\le\beta$ for some $i\in\{1,2\}$, then
$s'_i=r'_i$.

\item[Case 3:] If $\beta<r'_i\le2\beta$ for some $i\in\{1,2\}$, then
$s'_i=\beta$, except that $s'_i=\beta+1$ when $s_i=\alpha$ and
$r\equiv 3\mod 18$.

\item[Case 4:] If $2\beta<r'_i\le \floor{r/2}$ for some $i\in\{1,2\}$, then
$s'_i=\floor{r'_i/2}$.
\end{description}

Let $f_j=\min\{\FL{\FR{r-c_{3-j}}3},\bFL{\FR{r_j}{3-u_{3-j}}}\}$.
During the game, if the \revs\ swarm $X_{3-j}$ in the current round, then they
generate at most $f_j$ new meetings.  Hence it suffices to show that the
strategy specified above always ensures
$$
s_j\ge f_j\qquad \textrm{for}~j\in\{1,2\}.\eqno{(A)}
$$

As in Theorem~\ref{spym2}, in order to prove $(A)$ we will also need
$$
s_j\ge \alpha \qquad\textrm{for}~j\in\{1,2\}.\eqno{(B)}
$$
Place the spies to satisfy $(A)$ and $(B)$ in round $0$.  In each Case of
play, $\alpha\le s_i'\le \FL{r/4}\le s-\alpha$, so $(B)$ is preserved.
Now $s_1,s_2,s_1',s_2'\ge\alpha$, and we study $(A)$.

With $f'_j$ being the value of $f_j$ at the end of the round, we need
$s'_j\ge f'_j$.  By Remark~\ref{allmove}, each part receives at least $\alpha$
new spies in each round.  In Cases 2, 3, and 4 each part contains at least
$\alpha$ \revs, so $c'_j\ge \alpha$ in those cases.  Also
$s'_j\ge \floor{r'_j/3}$ in each Case.  Since
$s'_j\ge \floor{r'_j/3}=\floor{r'_j/(3-u'_{3-j})}$ when $u'_{3-j}=0$, 
we may assume $u'_j\in\{1,2\}$.

In addition, since the greedy strategy places new spies in $X_j$ to maximize
coverage, leaving an uncovered vertex with $u'_j$ \revs\ implies that each of
the (at least) $\alpha$ new spies covers at least $u'_j$ \revs\ at its vertex.
Hence $c'_j\ge u'_j \alpha$.

\medskip
\emph{Invariant (A) is preserved:} 

In Case 1, $s'_i=\alpha\ge r'_i\ge f'_i$ 
and $s'_{3-i}=s-\alpha\ge \floor{r/3}\ge f'_{3-i}$.

In Case 2, $s'_i=r'_i\ge f'_i$.
Also, $c'_i\ge \alpha$ and
$s'_{3-i}=s-r'_i\ge s-\beta=\FL{\FR{r-\alpha}3}\ge
\bFL{\FR{r-c'_i}3}\ge f'_{3-i}$.

In Case 3, then $c'_i\ge \alpha$.  In the nonexceptional case,
$s'_{3-i}=s-\beta=\FL{\FR{r-\alpha}3}\ge \bFL{\FR{r-c'_i}3}\ge f'_{3-i}$.
If $s_i=\alpha$ and $r\equiv 3\mod 18$, then $s'_{3-i}=s-\beta-1$ and we must
be a bit more careful.  Since all $\alpha$ spies that were in $X_i$ move to
$X_{3-i}$, and $r_i'\ge\beta+1$, we have $c_i'\ge\beta+1$, and hence
$\FL{\FR{r-\alpha}3}-1\ge\bFL{\FR{r-c'_i}3}$.

In Case 3 or Case 4,
if $u'_{3-i}=1$, then $s'_i\ge\FL{r'_i/2}=\bFL{\FR{r'_i}{3-u'_{3-i}}}\ge f'_i$.
If $u'_{3-i}=2$, then $c'_{3-i}\ge 2\alpha$.  Hence
$s'_i\ge\beta\ge\bFL{\FR{r-2\alpha}3}\ge\bFL{\FR{r-c'_{3-i}}3}\ge f'_i$, with
the exception that $\beta=\bFL{\FR{r-2\alpha}3}-1$ when $r\equiv 3\mod 18$.
In this case either $s'_i>\beta$, which suffices, or $s_i>\alpha$.  If
$s_i>\alpha$, then $X_{3-j}$ has more than $\alpha$ new spies, so
$c'_{3-i}\ge 2\alpha+2$, which fixes the problem for $r\equiv 3\mod 18$.

In Case 4, if $u'_i=1$, then $s'_{3-i}=s-\bFL{\FR{r'_i}2}\ge
\bFL{\FR{r'_{3-i}}2}=\bFL{\FR{r'_{3-i}}{3-u'_i}}\ge f'_{3-i}$.
If $u'_i=2$, then $c'_i\ge 2\alpha$.  Now
$s'_{3-i}=s-\bFL{\FR{r'_i}2}\ge \FL{\FR r2}-\bFL{\FR{\FL{r/2}}2}
=\bigceil{\FR{\FL{r/2}}2}\ge\bFL{\FR{r-2\alpha}3}\ge\bFL{\FR{r-c'_i}3}\ge
f'_{3-i}$.
\end{proof}

\begin{thm}\label{spymgen}
If $G$ is an $r$-large complete bipartite graph, 
then $\sigma(G,m,r)\le (1+{1\over \sqrt{3}}){r\over m}+1$.
\end{thm}
\begin{proof}
For $s\ge (1+{1\over \sqrt{3}}){r\over m}+1$, we present a greedy migration
strategy for $s$ spies that keeps the \revs\ from winning by swarming.  Suppose
first that $\FR rm<\FR1{1-1/\sqrt3}<2.5$.  In this case, the \revs\ can never
make more than two meetings.  We want to show that at most $4.75$ spies suffice.
In fact, four spies always suffice, because they can always arrange to keep two
spies on each side to handle up to two new meetings on the other side.  The
greedy migration strategy that always sets $s_1=s_2=2$ accomplishes this.
Henceforth, we may assume $\FR rm\ge\FR1{1-1/\sqrt3}$.

As usual, $r_j$ and $s_j$ count the \revs\ and spies in $X_j$ to begin a round,
$r'_j$ counts the \revs\ after they move, and $s'_j$ is the number of spies to
be computed for $X_j$ to end the round.  To determine $s'_1$ and $s'_2$, the
spies compute $x$, $\alpha$, $u_1$, and $u_2$ (not necessarily integers) such
that 
\begin{align}
x&\le\FL{r/m},\qquad x+{r/m}+1\le s,\qquad\textrm{and}\label{eq:definexalpha}\\
\alpha &= x+r/m-{r-u_1x\over m} = x+r/m-{r'_2\over m-u_1} = {r'_1\over m-u_2} = {r-u_2 x\over m}.
\label{eq:defineus}
\end{align}
We will show that such numbers always exist.
Now $s'_1$ and $s'_2$ are computed as follows:
\begin{description}
\item [Case 1:] If $\alpha\le x$, then $s'_1=\ceil{x}$ and $s'_2=s-s'_1$.
\item [Case 2:] If $\alpha>\floor{r/m}$, then $s'_1=\floor{r/m}$ and $s'_2=s-s'_1$.
\item [Case 3:] If $x < \alpha \le\floor{r/m}$, then $s'_1=\ceil{\alpha}$ and $s'_2=s-s'_1$.
\end{description}
Since always $s'_j\ge x$, greedy migration moves at least $\ceil{x}$ new spies
to each part in each round, by Remark~\ref{allmove}.  Consider a swarm.  If all
uncovered vertices in $X_j$ have at most $u_j$ \revs, then swarming $X_j$
generates at most $r'_{3-j}/(m-u_j)$ new meetings.  If some uncovered vertex in
$X_j$ has more than $u_j$ \revs, then by greedy migration at least $x$ spies in
$X_j$ have covered more than $u_j$ \revs\ each, and swarming $X_j$ forms at
most $(r-u_jx)/m$ new meetings.  Hence swarming $X_j$ fails to win if
\begin{equation}\label{lastswarm}
s'_{3-j}\ge \max\left\{\FR{r'_{3-j}}{m-u_j},\FR{r-u_jx}m\right\}.
\end{equation}

For $j=2$, both quantities on the right in (\ref{lastswarm}) equal $\alpha$, so
the condition is equivalent to $s'_1\ge \alpha$, which holds in Cases 1 and 3.
In Case 2, $s_1'=\FL{r/m}$, which always protects against swarming $X_2$ since
at most $\FL{r/m}$ meetings can be made.

For $j=1$, both quantities on the right in (\ref{lastswarm}) equal
$x+r/m-\alpha$, so the condition is equivalent to $s'_2\ge x+r/m-\alpha$.
Since $s-1\ge x+r/m$, proving $s'_2\ge s-1-\alpha$ shows that swarming $X_1$ is
ineffective.  In Case 1, $s'_2> r/m$, which suffices.  In Case 2 or 3,
$s'_1\le \CL{\alpha}$, so $s'_2\ge s-\CL{\alpha}>s-1-\alpha$, as desired.

\medskip
It remains to show that such numbers exist.  Solving \eqref{eq:defineus} yields
\begin{align*}
x   &= {\sqrt{9r^2+12r'_1r-12{r'_1}^2} \over 6m} \\
u_1 &= {r+mx-\sqrt{r^2+2rxm+x^2m^2-4xr'_1m} \over 2x}\text{ and }  \\
u_2 &= {r+mx-\sqrt{r^2-2rxm+x^2m^2+4xr'_1m} \over 2x}.
\end{align*}
Since $x\le r/(\sqrt{3}m)$, the inequalities in \eqref{eq:definexalpha} hold
when $\FR rm\ge\FR1{1-1/\sqrt3}$.
\end{proof}

\bigskip
\centerline{\bf Acknowledgment}
We thank the referees for careful reading and many helpful suggestions.

\end{document}